\documentclass[11pt]{article}

\usepackage[
  bookmarks=true,
  bookmarksnumbered=true,
  bookmarksopen=t rue,
  pdfborder={0 0 0},
  breaklinks=true,
  colorlinks=true,
  linkcolor=black,
  citecolor=black,
  filecolor=black,
  urlcolor=black,
]{hyperref}

\usepackage{amsthm}
\usepackage{subfigure}
\usepackage{amsmath}
\usepackage{amsfonts}
\usepackage{amssymb}
\usepackage{graphicx}
\usepackage{stmaryrd}
\usepackage[capitalise]{cleveref}
\usepackage{tikz}
\usepackage{etoolbox}
\usepackage{ifthen}
\usepackage{enumitem}

\usepackage{wasysym}

\usepackage[margin=1in]{geometry}

\newlist{parts}{enumerate}{1}
\crefname{partsi}{Part}{Parts}
\setlist[parts,1]{label=\arabic*.,ref=\arabic*}

\newtheorem{definition}{Definition}
\newtheorem{lemma}{Lemma}
\newtheorem{theorem}{Theorem}
\newtheorem{remark}{Remark}

\newcommand{\existsv}{\exists{\veee}}
\newcommand{\knownf}[1]{\ensuremath{F\node{#1}}}
\newcommand{\Majvals}[1]{\ensuremath{\mathit{Maj}\node{#1}}}
\newcommand{\eqdef}{\triangleq}
\newcommand{\OptMaj}{\mbox{$\mbox{\sc Opt}_{\Maj}$}}

\newcommand{\Crash}{\mathsf{Crash}}

\newcommand{\defemph}[1]{\textbf{\textit{#1}}}
\newcommand{\sat}{\models}

\newcommand{\exv}{\exists\valv}
\newcommand{\exZ}{\exists 0}
\newcommand{\decidei}{\mathsf{decide_i}}
\newcommand{\decideiZ}{\mathsf{decide_i(0)}}
\newcommand{\decideiO}{\mathsf{decide_i(1)}}

\newcommand{\gammacr}{\gamma^{\tee}_{\mathrm{cr}}}
\newcommand{\Proc}{\mathsf{Procs}}
\newcommand{\veee}{\mathtt{v}}
\newcommand{\tee}{\,\defemph{t}}

\newcommand{\Pz}{P_0}
\newcommand{\UPz}{\mbox{$\mbox{{\sc u-}}P_0$}}
\newcommand{\Maj}{\mathsf{Maj}}
\newcommand{\OptO}{\mbox{$\mbox{{\sc Opt}}_1$}}
\newcommand{\OptZ}{\mbox{$\mbox{{\sc Opt}}_0$}}
\newcommand{\OptZs}{\mbox{$\mbox{{\sc Opt}}^\mathrm{std}_0$}}
\newcommand{\UOptZ}{\mbox{$\mbox{{\sc u-Opt}}_0$}}

\newcommand{\sfa}{\alpha}
\newcommand{\Fmodel}{{\cal F}}
\newcommand{\CG}{{\cal G}}
\newcommand{\Vals}{{\tt V}}
\newcommand{\Vecs}{\vec{\Vals}}
\newcommand{\dom}{\,{\preceq}\,}
\newcommand{\hmwopt}{P0_{\mathrm{opt}}}
\newcommand{\FP}{\mathsf{F}}
\newcommand{\notnz}{\mathsf{not\hbox{-}known}(\exists 0)}
\newcommand{\node}[1]{\langle#1\rangle}
\newcommand{\Ga}{\CG_\alpha}
\newcommand{\dec}{\mathsf{d}}
\newcommand{\cv}{\exists\mathsf{correct}(\veee)}
\newcommand{\cz}{\exists\mathsf{correct}(0)}

\newcommand{\Agreement}{{\bf Agreement}}
\newcommand{\UniAg}{{\bf Uniform Agreement}}
\newcommand{\Decision}{{\bf Decision}}
\newcommand{\Validity}{{\bf Validity}}
\newcommand{\valv}{\veee}
\newcommand{\fip}{{\it fip}}
\newcommand{\noactnv}{\mathsf{no\mbox{-}decided}(\bar\valv)}
\newcommand{\noactnO}{\mathsf{no\mbox{-}decided}(0)}
\newcommand{\KoP}{\mbox{{\bf K}$\!\!\!\!\;\;${\it o}{\bf P}}\/}
\newcommand{\val}{v}

\hypersetup{
  pdfauthor      = {Armando Casta{\~n}eda <armando@cs.technion.ac.il>, Yannai A. Gonczarowski <yannai@gonch.name>, and Yoram Moses <moses@ee.technion.ac.il>},
  pdftitle       = {Unbeatable Consensus}, 
}

\begin{document}
 
\title{Unbeatable Consensus\thanks{Part of the results of this paper were announced in~\cite{AYY-PODC-BA}.}}

\author{Armando Casta\~{n}eda\thanks{Universidad Nacional Aut\'onoma de M\'exico (UNAM), \mbox{\emph{E-mail}: \href{mailto:armando@cs.technion.ac.il}{armando@cs.technion.ac.il}}.}
\and Yannai A.~Gonczarowski\thanks{The Hebrew University of Jerusalem and Microsoft Research, \mbox{\emph{E-mail}: \href{mailto:yannai@gonch.name}{yannai@gonch.name}}.}
\and Yoram Moses\thanks{Technion, \mbox{\emph{E-mail}: \href{mailto:moses@ee.technion.ac.il}{moses@ee.technion.ac.il}}.}
}

\date{September 15, 2014}

\maketitle

\begin{abstract}
The \defemph{unbeatability} of a consensus protocol, introduced by Halpern, Moses and Waarts in~\cite{HalMoWa2001}, 
is a stronger notion of optimality than the accepted notion of early stopping protocols. 
Using a novel knowledge-based analysis, this paper derives the first practical unbeatable consensus protocols in the literature, for the standard synchronous message-passing model with crash failures. These protocols strictly dominate the best known protocols for uniform and for non-uniform consensus, in some case beating them by a large margin. 
The analysis provides a new understanding of the logical structure of consensus, and of the distinction between uniform and nonuniform consensus. Finally, the first (early stopping and) unbeatable protocol that treats decision values ``fairly'' is presented. All of these protocols have very concise descriptions, and  are shown to be efficiently implementable.

\vspace{1em}

\end{abstract}

\section{Introduction}

Following \cite{HMT11}, we say that a protocol~$P$ is a \defemph{worst-case optimal} solution to a decision task~$S$ in a given model 
if it solves~$S$, and decisions in~$P$ are always taken no later than the {\em worst-case} lower bound for decisions in this problem, in that model.
Here we consider standard synchronous message-passing models with $n$ processes and
at most $\tee < n$ crash failures per run; it will be convenient to denote the number of \emph{actual} failures in a given run by~$f$. Processes proceed in a sequence of synchronous rounds.
The very first consensus protocols were worst-case optimal, deciding in exactly $\tee+1$ rounds in all runs \cite{DS,PSL}.
It was soon realized, however, that they could be strictly improved upon by \defemph{early stopping} protocols~\cite{DRS}, 
which are also worst-case optimal, but can often decide much faster than the original ones. 
This paper presents a number of consensus
protocols that are not only worst-case optimal and early stopping, but furthermore cannot be strictly improved upon, 
and are thus optimal in a much stronger sense. 

In benign failure models it is typically possible to define the behaviour of the environment (i.e., the adversary) in a manner that is independent of the protocol, in terms of a pair $\alpha=(\vec{v},\FP)$ consisting of a vector $\vec{v}$ 
of initial values and a failure pattern~$\FP$. (A formal definition is given in Section~\ref{sec:model}.)
A failure model~$\Fmodel$  is identified with a set of (possible) failure patterns. 
For ease of exposition, we will think of such a pair $\alpha=(\vec{v},\FP)$ as a particular {\em adversary}. 
In a synchronous environment, 
a deterministic protocol~$P$ and an adversary~$\alpha$ uniquely define a run $r=P[\alpha]$. 
With this terminology, we can compare the performance of different decision protocols solving a particular task in a given context $\gamma=(\Vecs,\Fmodel)$, where $\Vecs$ is a set of possible vectors of initial values.
A decision protocol $Q$ \defemph{dominates} a protocol~$P$ in~$\gamma$, denoted by $Q\boldsymbol{\dom_\gamma} P$ if, for all adversaries $\alpha$ and every process~$i$, if $i$ decides in~$P[\alpha]$ at time 
$m_i$, then $i$ decides in $Q[\alpha]$ at some time 
$m'_i\le m_i$. Moreover, we say that $Q$  \defemph{strictly dominates} $P$
if $Q\dom_\gamma P$ and  $P\!\!\boldsymbol{\not}\!\!\!\dom_\gamma Q$. I.e., if~$Q$ dominates~$P$ and for some $\alpha\in\gamma$ there exists a process~$i$ that decides in $Q[\alpha]$ {\em strictly before} it does so in $P[\alpha]$. 
In the crash failure model, the early-stopping protocols of \cite{DRS} strictly dominate the original protocols of \cite{PSL}, in which decisions are always performed at time $\tee+1$. Nevertheless, these early stopping protocols may not be optimal solutions to consensus. 
Following \cite{HMT11} a protocol~$P$ is said to be an  \defemph{all-case optimal} solution to a decision task~$S$ in a context~$\gamma$ if it solves~$S$ and, moreover, it dominates every protocol~$P'$ that solves~$S$ in~$\gamma$. 
Dwork and Moses presented all-case optimal solutions to the {\em simultaneous} variant of consensus
\cite{DM}.
 For the standard ({\em eventual}) variant of consensus, in which decisions 
 are not required to 
 occur simultaneously, Moses and Tuttle showed that no all-case optimal solution exists~\cite{MT}. 
Consequently, Halpern, Moses and Waarts in \cite{HalMoWa2001} initiated the study of a natural notion of optimality 
that is achievable by eventual consensus protocols:

\begin{definition}[Halpern, Moses and Waarts]
A protocol $P$ is an \defemph{unbeatable} solution to a decision task~$S$ in a context~$\gamma$ if $P$ solves~$S$ in~$\gamma$ and no protocol $Q$ solving~$S$ in~$\gamma$ strictly dominates~$P$.%
\footnote{All-case optimal protocols are called {\em ``optimal in all runs''} in \cite{DM}. 
 They are termed  {\em ``optim\defemph{um}''} in \cite{HalMoWa2001}, while unbeatable protocols are simply called {\em ``optim\defemph{al}''} there. We prefer the term {\em unbeatable} because ``optimal'' is used very broadly, and inconsistently, in the literature.}
\end{definition}

Halpern, Moses and Waarts 
observed that for every consensus protocol $P$ there exists an unbeatable protocol $Q_P$ that dominates~$P$. Moreover, they showed a two-step transformation that defines such a protocol~$Q_P$ based on~$P$. 
This transformation and the resulting protocols are based on a notion of {\em continual} common knowledge that is computable, but not 
efficiently: 
in the resulting protocol, each process executes exponential time (PSPACE) local 
computations in every round. The logical transformation is not applied in \cite{HalMoWa2001} to an actual protocol. 
As an example of an unbeatable protocol, they present a particular protocol, called~$\hmwopt$, and argue that it is unbeatable in the crash failure model.
Unfortunately, as we will show,  $\hmwopt$ is in fact beatable. This does not refute the general analysis and transformation defined in \cite{HalMoWa2001}; they remain correct. Rather, the fault is in an unsound step in the proof of optimality of $\hmwopt$ (Theorem~{6.2} of \cite{HalMoWa2001}), in which an inductive step is not explicitly detailed, and does not hold.

\vspace{1em}

The main contributions of this paper are:

\begin{enumerate}
\item 
A knowledge-based analysis is applied to the classical consensus protocol, and is shown to yield solutions that are optimal in a much stronger sense than all previous solutions. Much simpler and more intuitive than the framework used in \cite{HalMoWa2001}, it illustrates how the knowledge-based approach can yield a structured approach to the derivation of efficient protocols. 

\item $\OptZ$, the first explicit unbeatable protocol for nonuniform consensus is presented. It is computationally efficient, and its unbeatability is established by way of a succinct proof. Moreover,  $\OptZ$ is shown to strictly dominate the $\hmwopt$ protocol from \cite{HalMoWa2001}, proving that the latter is in fact beatable.
\item An analysis of uniform consensus gives rise to~$\UOptZ$, the first explicit unbeatable protocol for uniform consensus. The analysis used in the design of~$\UOptZ$ sheds light on the inherent difference and similarities between the uniform and nonuniform variants of consensus in this model. 

\item 
Early stopping protocols for consensus are traditionally one-sided,  preferring to decide on~0 (or on~1) if possible.  deciding on a predetermined value (say, 0) if possible, we present an 
An unbeatable (and early stopping) majority consensus protocol~$\OptMaj$  is presented, that prefers the \emph{majority} value.
\item We identify the notion of a {\em hidden path} as being crucial to decision in the consensus task. 
If a process identifies that no hidden path exists, then it can decide. 
In the fastest early-stopping protocols, a process decides after the first round in which it does not detect a new failure. 
By deciding based on the nonexistence of a hidden path, our unbeatable protocols can stop up to $\tee-3$ rounds  faster than the best early stopping protocols in the literature. 
\end{enumerate}

We  now sketch the intuition behind, our unbeatable consensus protocols. 

In the standard version of consensus,  every process~$i$ starts with an initial value $v_i\in\{0,1\}$, and the following properties must hold in every run~$r$: 

\vspace{\topsep}
\noindent
\underline{~(Nonuniform) {\bf  Consensus:~}}
\begin{itemize}
\item[]{\bf Decision:}\quad Every correct process must decide on some value, 
\item[]{\bf Validity:}\quad If all initial values are~$\valv$ then the correct processes decide~$\valv$, and 
\item[]{\bf Agreement:}\quad All correct processes decide on the same value.
\end{itemize}

The connection between knowledge and distributed computing was proposed in~\cite{HM1} and has been used in the analysis of a variety of problems, including consensus (see \cite{FHMV} for more details and references). 
In this paper, we employ simpler techniques to perform a more direct knowledge-based analysis.
Our approach is based on a simple principle recently formulated by Moses in~\cite{Mono}, called the \defemph{knowledge of preconditions} principle (\KoP), which captures an essential connection between knowledge and action in distributed and multi-agent systems. 
Roughly speaking, the \KoP\ principle says that 
if~$C$ is a necessary condition for an action~$\sfa$ to be performed by process~$i$, then $K_i(C)$ --- $i$~knowing~$C$ --- is a  necessary condition for~$i$ performing~$\sfa$. E.g., it is not enough for a client to have positive credit in order to receive cash from an ATM; the ATM must \defemph{know} that the client has positive credit. 

Problem specifications typically state or imply a variety of necessary conditions. 
In the crash failure model studied in this paper,  we will say that a process is \defemph{active} at time~$m$ in a given run, if it does not crash before time~$m$. For $\valv\in\{0,1\}$, we denote by 
${\decidei(\valv)}$ the action of~$i$ deciding~$\valv$, and use $\bar \valv$ as shorthand for $1-\valv$. 

\begin{lemma}
\label{lem:necessary}
Consensus implies the following necessary conditions for ${\decidei(\valv)}$ in the crash failure model: 
\begin{enumerate}
\item[(a)] ``at least one processes had initial value~$\valv$'' (we denote this by~$\boldsymbol{\existsv}$), and 
\item[(b)] ``no currently active process has decided, or is currently deciding, 
$\bar\valv$''
(we denote this by $\noactnv$).
\end{enumerate}
\end{lemma}

Both parts follow from observing that if~$i$ decides~$\valv$ at a point where either (a) or (b) does not hold, then the execution can be extended to a run in which $i$ (as well as $j$, for (b)) is correct (does not crash), and this run violates \Validity\ for (a) or \Agreement\ for (b). 

Given \cref{lem:necessary}, \KoP\ implies that  $K_i\existsv$ and $K_i\noactnv$ are also necessary conditions for $\decidei(\valv)$. 
In this paper, we will explore how this insight can be exploited in order to design efficient consensus protocols. 
Indeed, our first unbeatable protocol will be one in which, roughly speaking,  the rule for $\decideiZ$ will be 
$K_i\exZ$, and the rule for $\decideiO$ will be  $K_i\noactnO$. 
As we will show, if the rule for $\decideiZ$ is $K_i\exZ$, then $\noactnO$ reduces to the fact $\notnz$, 
which is true at a given time if $K_j\exZ$ holds for no currently-active process~$j$. 
Thus, $K_i\noactnO$ --- our candidate rule for deciding~1 --- then becomes $K_i\notnz$. 
While $K_i\exZ$ involves the knowledge a process has about initial values, $K_i\notnz$ is concerned with~$i$'s knowledge about the knowledge of others. We will review the formal definition of knowledge in the next section, in order to turn this into a rigorous condition. 

Converting the above description into an actual protocol essentially amounts to providing concrete tests for when these knowledge conditions hold. 
 It is straightforward to show (and quite intuitive) that in a full-information protocol $K_i\exZ$ holds exactly if there is a message chain from some process~$j$ whose initial value is~0, to process~$i$. 
To determine that $\notnz$, a process must have proof that no such chain can exist. 
Our technical analysis identifies a notion of a {\em hidden path} with respect to~$i$ at a time~$m$, which implies that a 
message chain 
could potentially be communicating a value unbeknownst to~$i$. 
It is shown that hidden paths are key to evaluating whether $K_i\notnz$ holds. 
In fact, it turns out that hidden paths are key to obtaining additional  unbeatable protocols in the crash failure model. We present two such protocols; one is a consensus protocol in which a process that sees a majority value can decide on this value, and the other is an unbeatable protocol for the {\em uniform} variant of consensus. In uniform consensus, 
any two processes that decide must decide on the same value, even if one (or both) of them crash soon after deciding. 
 
This paper is structured as follows: 
The next section reviews the definitions of the synchronous crash-failure model and of knowledge in this model. 
\cref{sec:PA-con} presents~$\OptZ$, our unbeatable consensus protocol, proves its unbeatability, and shows that it beats the protocol~$\hmwopt$ of~\cite{HalMoWa2001}.
It then derives an  unbeatable consensus protocol, $\OptMaj$, that treats~0 and~1 in a balanced way.
Both unbeatable protocols decide in no more than~$f+1$ rounds in runs in which~$f$ processes actually fail 
but they can decide much earlier than that. 
\cref{sec:uniform} studies uniform consensus, and derives $\UOptZ$, an unbeatable protocol for uniform consensus. 
Finally, \cref{sec:discussion} concludes with a discussion. The Appendix contains full proofs to all claims that are not proved in the main text. 

\section{Preliminary Definitions}
\label{sec:model}

Our model of computation is the standard synchronous message-passing model with 
benign crash failures.
A system has~\mbox{$n\!\ge\!2$} processes denoted by  
$\Proc=\{1,2,\ldots,n\}$. 
Each pair of processes is connected by a two-way communication link,
and each message is tagged with the identity of the sender.
They share a discrete global clock that starts out at time~$0$ and
advances by increments of one. Communication in the system proceeds in
a sequence of \emph{rounds}, with round~$m+1$ taking place between
time~$m$ and time~$m+1$.
Each process starts in some \emph{initial state} at time~$0$,
usually with an \emph{input value} of some kind.
In every round, each process first performs a local computation, and performs local actions,
then  
it sends a set of messages to other processes, and finally receives messages sent to it
by other processes during the same round. 
We consider the 
local computations and sending actions of round~$m+1$ as being performed at time~$m$, 
and the messages are received at time~$m+1$.

A faulty process fails by \emph{crashing} in some round~$m\ge 1$. 
It behaves correctly in the first~$m-1$ rounds and 
sends no messages from round~$m+1$ on. 
During its crashing round~$m$, the process may succeed in
sending messages on an arbitrary subset of its links. 
At most~$\tee \leq n-1$ processes fail in any given execution.

It is convenient to consider the state and behaviour of processes at different 
(process-time) nodes, where a \defemph{node}  is a pair $\node{i,m}$ referring to process~$i$ at time~$m$.
A \defemph{failure pattern} describes how processes fail in an execution.
It is a layered graph~$\FP$ whose vertices are 
nodes~$\node{i,m}$
 for $i\in\Proc$ and $m\ge 0$. 
Such a vertex denotes process~$i$ and time~$m$. 
An edge has the form $(\node{i,m-1},\node{j,m})$ 
and it denotes the fact that a message sent by~$i$ to~$j$ in round~$m$ would be delivered successfully. 
Let~$\Crash(\tee)$ denote the set of failure patterns in which all failures are crash failures, and no more than~$\tee$ crash failures occur. 
An \defemph{input vector} describes the initial values that the processes receive in an 
execution. The only inputs we consider are initial values that processes obtain at time~0. 
An input vector is thus a tuple $\vec{v}=(v_1,\ldots,v_n)$ where~$v_j$ is the input to process~$j$. 
We think of the input vector and the failure pattern as being determined by an external scheduler, and thus a  pair $\alpha=(\vec{v},\FP)$ is called an {\em adversary}. 

A \defemph{protocol} describes what messages a process sends and what decisions it takes, 
as a deterministic function
 of its local state
at the start of a round and the messages received during a round.
We assume that a protocol~$P$ has access to the values of~$n$ and~$\tee$,
typically passed to~$P$ as parameters.

A \defemph{run} is a description of an infinite behaviour of the system.
Given a run~$r$ and a time~$m$, 
we denote by $r_i(m)$
the \defemph{local state} of process~$i$ at time~$m$ in~$r$,
and the \defemph{global state} at time $m$ 
is defined to be $r(m)=\node{r_1(m),r_2(m),\ldots,r_n(m)}$.
A protocol~$P$ and an adversary~$\alpha$ uniquely determine a run, 
and we write $r = P[\alpha]$.

Since we restrict attention to benign failure models and focus on decision times and solvability in this paper, it is sufficient to consider {\em full-information} protocols ({\em fip}'s for short), defined below \cite{Coan}. 
There is a convenient way to consider such protocols in our setting. 
With an adversary $\alpha=(\vec{v},\FP)$ we associate a \defemph{communication graph} $\CG_\alpha$, 
consisting of the graph~$\FP$ extended by labelling the initial nodes $\node{j,0}$ with the initial states $v_j$ according to~$\alpha$. 
Every node $\node{i,m}$ is associated with a subgraph  $\Ga(i,m)$ of~$\CG_\alpha$, which we think of as $i$'s {\em view} at $\node{i,m}$.
Intuitively, this graph will represent all nodes $\node{j,\ell}$ from which $\node{i,m}$ has heard, and the initial values it has seen. 
Formally, $\Ga(i,m)$ is defined by induction on~$m$. 
$\Ga(i,0)$ consists of the node $\node{i,0}$, labelled by the initial value~$v_i$. 
Assume that $\Ga(1,m),\ldots,\Ga(n,m)$ have been defined, and let $J\subseteq\Proc$ be the set of processes~$j$ such that $j=i$ or $e_j=(\node{j,m},\node{i,m+1})$ is an edge of~$\FP$. Then $\Ga(i,m+1)$ consists of the node $\node{i,m+1}$, the union of all graphs $\Ga(j,m)$ with $j\in J$, and the edges 
$e_j=(\node{j,m},\node{i,m+1})$ for all $j\in J$. 
We say that $(j,\ell)$ is {\em seen} by $\node{i,m}$ if $(j,\ell)$ is a node of $\Ga(i,m)$. Note that this occurs exactly if the failure pattern~$\FP$ allows a (Lamport) message chain from $\node{j,\ell}$ to $\node{i,m}$. 

A full-information protocol $P$ is one in which at every node $\node{i,m}$ of a run $r=P[\alpha]$ the process~$i$ constructs $\Ga(i,m)$ after receiving its round~$m$ nodes, and sends $\Ga(i,m)$ to all other processes in round~$m+1$. In addition, $P$ specifies what decisions $i$ should take at $\node{i,m}$ based on $\Ga(i,m)$.
Full-information protocols thus differ only in the decisions taken at the nodes. 
Let $\dec(i,m)$ be 
status of~$i$'s decision at time~$m$ (either~`$\bot$' if it is undecided, or a concrete value~`$\veee$').
Thus, in a run $r=P[\alpha]$, we define the local state 
$r_i(m) = \langle \dec(i,m),\Ga(i,m)\rangle$ if~$i$ does not crash before time~$m$ according to~$\alpha$, and $r_i(m)=\frownie$, an uninformative ``crashed'' state,  if~$i$  crashes
before time~$m$.

For ease of exposition and analysis, all of our  protocols are full-information. 
However, in fact, they can all be implemented in such a way that any process sends any other 
process a total of $O(f \log n)$ bits throughout any execution  
(as shown by Lemma~\ref{nlogn} in Appendix~\ref{sec-notions}).

\subsection{Knowledge}

Our construction of unbeatable protocols will be assisted and guided by a knowl\-edge-based analysis, in the spirit of \cite{FHMV,HM1}. 
Runs are dynamic objects, changing from one time point to the next. E.g., at one point process~$i$ may be undecided, while at the next it may decide on a value. Similarly, the set of initial values that~$i$ knows about, or has seen, may change over time. In general, whether a process ``knows'' something at a given point can depend on what is true in other runs in which the process has the same information. 
We will therefore consider the truth of facts at {\em points} $(r,m)$---time~$m$ in run~$r$, with respect to a 
set of runs~$R$ 
(which we call a \defemph{system}). 
We will be interested in systems of 
the form $R_P=R(P,\gamma)$ where $P$ is a protocol and $\gamma=\gamma(\Vals^n,\Fmodel)$ is the set of all adversaries that assign initial values from~$\Vals$ and failures according to~$\Fmodel$. We will write $(R,r,m)\sat A$ to state that fact~$A$ holds, or is satisfied, at $(r,m)$ in the system~$R$.

The truth of some facts can be defined directly. 
For example, the fact $\existsv$ will hold at $(r,m)$ in~$R$
if some process has initial value~$\veee$ in~$(r,0)$. We say that {\em (satisfaction of)} a fact~$A$ is \defemph{well-defined in~$R$} if 
for every point $(r,m)$ with $r\in R$ we can determine whether or not $(R,r,m)\sat A$. 
Satisfaction of~$\existsv$ is thus well defined. 
Moreover, any boolean combination of well-defined facts is also well defined.
We will write $K_iA$ to denote that \defemph{process~$i$ knows~$A$}, and define: 

\begin{definition}[Knowledge]
\label{def:know}
Suppose that~$A$ is well defined in~$R$. Define that 

\noindent\begin{tabular}{r l c l}
$(R,r,m)\sat K_iA$ & ~~iff~~ & $(R,r',m)\sat A$ 
\mbox{holds for all} $r'\in R$ \mbox{with}
$r_i(m)=r'_i(m)$.
\end{tabular}
\end{definition}

Thus, if $A$ is well defined in~$R$ then Definition~\ref{def:know} makes $K_iA$ well defined in~$R$. 
Note that what a process knows or does not know depends on its local state.
The definition can then be applied recursively, to define the truth of $K_jK_iA$ etc. Knowledge has been used to study a variety of problems in distributed computing. 
In particular, we now formally define 
$(R,r,m)\sat\notnz$ to hold iff
\mbox{$(R,r,m)\not\sat K_j\exists 0$}~~holds for 
every process~$j$ that does not crash by 
time~$m$ in~$r$.
We will make use of the following  fundamental connection between knowledge and actions in distributed systems. 
A fact~$A$ is a \defemph{necessary condition} for process~$i$ performing action~$\sigma$ 
(e.g. deciding 
on an output value)
in~$R$ if 
$(R,r,m)\sat A$ whenever $i$ performs $\sigma$ at a point $(r,m)$ of~$R$. 

\begin{theorem}[Knowledge of Preconditions, \cite{Mono}]
\label{thm:knowprec}
Let $R_P=R(P,\gamma)$ be the set of runs of a deterministic protocol~$P$. 
If $A$ is a necessary condition for~$i$ performing~$\sigma$ in~$R_P$, then so is $K_iA$.
\end{theorem}

\section{Unbeatable Consensus}
\label{sec:PA-con}

We start with the standard version of consensus defined in the Introduction, and consider the crash failure context $\gammacr=\langle\Vals^n,\Crash(\tee)\rangle$, where $\Vals=\{0,1\}$ --- initial values are binary bits. Every protocol~$P$ in this setting determines a system $R_P=R(P,\gammacr)$. 
Recall that \cref{lem:necessary} establishes necessary conditions for decision in consensus. Based on this, 
Theorem~\ref{thm:knowprec} yields:

\begin{lemma}\label{lem:know-exists}
Let $P$ be a consensus protocol for~$\gammacr$ and let $R_P=R(P,\gammacr)$. Then both $K_i\existsv$ and $K_i\noactnv$ are necessary conditions for $\decidei(\valv)$ in~$R_P$. 
\end{lemma} 

An analysis of knowledge for {\fip}s in the crash failure model was first performed  by Dwork and Moses in \cite{DM}. 
The following result is an immediate consequence of that analysis. 
Under the full-information protocol, we have: 

\begin{lemma}[Dwork and Moses~\cite{DM}]
\label{lem:knowing0}
Let $P$ be a \fip\ in~$\gammacr$ and let $r\in R_P=R(P,\gammacr)$. 
For all processes $i,j$,
~$(R_P,r,\tee+1)\sat K_i\exv$ ~iff~ \mbox{$(R_P,r,\tee+1)\sat K_j\exv$}.
\end{lemma}

Of course, a process that does not know $\exZ$ must itself have an initial value of~1. Hence, based on \cref{lem:knowing0}, it is natural to design a \fip-based  consensus protocol that performs $\decideiZ$ at time~$\tee+1$ if~$K_i\exZ$, and otherwise performs ~$\decideiO$. (In the very first consensus protocols, all decisions are performed at time \mbox{$\tee+1$}~\cite{PSL}.) 
Indeed, one can use  \cref{lem:knowing0} to obtain a strictly better protocol, 
in which decisions on~0 are performed sooner: 

\vspace{\topsep}
\noindent
~~\underline{~{\bf Protocol}~$\Pz$~}
(for an undecided process~$i$ at time~$m$):\\
\begin{tabular}{lll}
\qquad\qquad{\bf if} &  $K_i\exists{0} $  & {\bf then} $\decideiZ$\\
\qquad\qquad{\bf  elseif} &  $m=\tee+1$
& {\bf then} $\decideiO$
\end{tabular}
\vspace{\topsep}

Notice that in a \fip\ consensus protocol, it is only necessary to describe the rules for $\decideiZ$ and $\decideiO$,
since in every round a process sends all it knows to all processes.
Since $K_i\exZ$ is a necessary condition for $\decideiZ$, the protocol $\Pz$ decides on~0 as soon  as any consensus protocol can.  In the early 80's Dolev suggested a closely related protocol~$B$ (standing for {\em ``Beep''}) for $\gammacr$, in which processes decide~0 and broadcast the existence of a~0 when they see a~0, and decide~1 at~$\tee+1$ otherwise~\cite{DolevBeep}; for all adversaries, it performs the same decisions at the same times as~$\Pz$. 
Halpern, Moses and Waarts show in~\cite{HalMoWa2001} that for every consensus protocol~$P$ in~$\gammacr$  
there is an unbeatable consensus protocol~$Q$ dominating~$P$. 
Our immediate goal is to obtain an unbeatable consensus protocol dominating~$\Pz$.
To this end, we make use of the following.

\begin{lemma}\label{lem:decide-when-0}
If $Q\dom\Pz$ is a consensus protocol, then~$\decideiZ$  
is performed  in~$Q$ exactly when $K_i\exists{0}$ first holds. 
\end{lemma}

We can now formalize the discussion in the Introduction, showing that if decisions on~0 are performed precisely when $K_i\exZ$ first holds, then $\noactnO$ reduces to~$\notnz$. 

\begin{lemma}\label{lem:decide-when-notnz}
Let $P$ be a fip,
in which $\decideiZ$  
is performed  in~$P$ exactly when $K_i\exists{0}$ first holds, and let $R_P=R(P,\gammacr)$. 
Then $(R_P,r,m)\sat K_i\noactnO$ ~iff~ $(R_P,r,m)\sat K_i\notnz$ ~for all $r\in R_P$ and $m\ge 0$.
\end{lemma}

The proof of \cref{lem:decide-when-notnz} is fairly immediate: If 
$(R_P,r,m)\not\sat K_i\notnz$ then there is a run~$r'$ of~$R_P$ such that both $r_i(m)=r'_i(m)$ and $(R_P,r',m)\sat K_j\exZ$ for 
some
correct process~$j$; therefore, process~$j$ decides~0 in~$r'$.
The other direction follows directly from the decision rule for $0$.
We can now define a \fip\ consensus protocol in which~0 is defined as soon as its necessary condition $K_i\exZ$ holds, and~1 is decided as soon as possible, given the rule for deciding~0:

\vspace{\topsep}
\noindent
~~\underline{~{\bf Protocol}~$\OptZ$~}
 (for an undecided process~$i$ at time~$m$):\\
\begin{tabular}{lll}
\qquad\qquad  {\bf if} & $K_i\exists{0}$  &  {\bf then} $\decideiZ$\\
\qquad\qquad  {\bf elseif} & $K_i\notnz$
&{\bf then} $\decideiO$
\end{tabular}
\vspace{\topsep}

We can show that~$\OptZ$ is, indeed,  an unbeatable protocol:

\begin{theorem}\label{thm:optz}
$\OptZ$ is an unbeatable consensus protocol in $\gammacr$. 
 \end{theorem}
 
\subsection{Testing for Knowing that Nobody Knows}

$\OptZ$ is not a standard protocol, because its actions depend on tests for process~$i$'s knowledge.
(It is a {\em knowledge-based program} in the sense of \cite{FHMV}.) 
In order to turn it into a standard protocol, we need to replace these by explicit tests on the processes' local states. The rule for~$\decideiZ$ is easy to implement. By \cref{lem:knowing0}(a), $K_i\exZ$ holds exactly if~$i$'s local state contains a time~0 node that is labelled with value~0. 
The rule $K_i\notnz$ for performing $\decideiO$ holds when~$i$ knows that no active process knows~$\exZ$, and we now characterize when this is true. A central role in our analysis will be played by process~$i$'s knowledge about the contents of various nodes in the communication graph. 
Recall that local states $r_i(m)$  in \fip's are communication graphs of the form~$\Ga(i,m)$; we abuse notation and 
write $\theta\in r_i(m)$ \big(respectively, $(\theta,\theta')\in r_i(m)$\big)  if~$\theta$ is a node of~$\Ga(i,m)=r_i(m)$ 
\big(respectively, if $(\theta,\theta')$ is an edge of~$\Ga(i,m)=r_i(m)$\big); in this case, we say that $\theta$ is \defemph{seen} by $\node{i,m}$.
 We now make the following definition: 

\begin{definition}[Revealed]
\label{def:revealed}
Let~$r\in R_P=R(P,\gammacr)$ for a \fip\ protocol~$P$.
We say that
\defemph{node~$\boldsymbol{\node{j',m'}}$  is revealed to~$\boldsymbol{\node{i,m}}$ in}~$\boldsymbol{r}$ if either ~(1) $\node{j',m'}\in r_i(m)$, or ~(2)  for some process~$i'$ such that  $\node{i',m'}\in r_i(m)$ it is the case that \mbox{$\big(\node{j',m'-1},\node{i',m'}\big)\notin r_i(m)$}. 
We say that \defemph{time~$\boldsymbol{m'}$ is revealed to~$\boldsymbol{\node{i,m}}$ in}~$\boldsymbol{r}$ if $\node{j',m'}$ is revealed to $\node{i,m}$ for all processes $j'$.
\end{definition}

Intuitively,  if node $\node{j',m'}$ is revealed to $\node{i,m}$ then~$i$ has proof at time~$m$ that $\node{j',m'}$ can not carry information that is not known at $\node{i,m}$ but may be known at another node $\node{j,m}$ at the same time. This because either~$i$ sees $\node{j',m'}$ at that point---this is part~(1)---or~$i$ has proof that~$j'$ crashed before time~$m'$, and so its state there was~$\frownie$, and~$j'$ did not send any messages at or after time~$m'$. 
It is very simple and straightforward from the definition to determine which nodes are revealed to~$\node{i,m}$, based on $r_i(m)=\Ga(i,m)$.
Observe that if a node~$\node{j',m'}$ is revealed to~$\node{i,m}$, then~$i$ knows at~$m$ what message could have been sent at~$\node{j',m'}$: 
If $\node{j',m'}\in r_i(m)$ then $r_{j'}(m')$ is a subgraph of $r_i(m)$, while if~$\big(\node{j',m'-1},\node{i',m'}\big)\notin r_i(m)$ for some node $\node{i',m'}\in r_i(m)$, then~$j'$ crashed before time~$m'$ in~$r$, and so it sends no messages at time~$m'$.  
Whether and when a node $\node{j',m'}$ is revealed to~$i$ depends crucially on the failure pattern. If~$i$ receives a message from~$j'$ in round~$m'+1$, then $\node{j',m'}$ is immediately revealed to 
 $\node{i,m'+1}$. If this message is not received by $\node{i,m'+1}$, then $\node{j',m'+1}$ --- the successor of $\node{j',m'}$ --- becomes revealed (as being crashed, i.e.\ in state~$\frownie$) to $\node{i,m'+1}$. But in general $\node{j',m'}$ can be revealed to~$i$ at a much later time than $m'+1$, (A simple instance of this is when $K_i\exZ$ first becomes true at a time $m>1$; this happens when $\node{j,0}$ with $\val_j=0$ is first revealed to~$i$.) 

 Suppose that some time $k\le m$ is revealed to~$\node{i,m}$. 
Then, in a precise sense, process~$i$ at time~$m$ has 
all of the information that existed in the system at time~$k$ (in the hands of processes that had not crashed by then).
In particular, if this information does not contain an initial value of~0, then nobody can know~$\exZ$ at or after time~$m$. 
We now formalize this intuition and show that revealed nodes can be used to determine when a process can know $\notnz$.

\begin{lemma}
\label{lem:rev}
Let~$P$ be a \fip\ and let~$r\in R_P=R(P,\gammacr)$. 
For every node $\node{i,m}$, it is the case that
$(R_P,r,m)\sat K_i\notnz$ exactly if both (1)~$(R_P,r,m)\not\sat K_i\exZ$ and ~(2)  
some time $k\le m$ is revealed to~$\node{i,m}$ in~$r$.
\end{lemma}

Based on \cref{lem:rev}, we now obtain a standard unbeatable consensus protocol for~$\gammacr$ that implements~$\OptZ$:

\vspace{\topsep}
\noindent
~~\underline{~{\bf Protocol}~$\OptZs$~}
 (for an undecided process~$i$ at time~$m$):\\
\begin{tabular}{lll}
\qquad\qquad  {\bf if} & 
$i$ has seen a time-$0$ node with initial value~$0$  
&  {\bf then} $\decideiZ$\\
\qquad\qquad  {\bf elseif} &   
some time $k\le m$ is revealed to~$\node{i,m}$
&{\bf then} $\decideiO$
\end{tabular}
\vspace{\topsep}

We emphasize that $\OptZs$ (and thus also $\OptZ$), and all the following protocols, can be implemented efficiently. The protocol only uses information about the existence of~0 and about the rounds at which processes crash. 
It can therefore be implemented in such a way that any process sends a total of $O(f\log n)$ bits (see Lemma~\ref{nlogn} in Appendix~\ref{sec-notions}) in every run, 
and executes $O(n)$ local steps in every round.

The formulation of $\OptZs$, in addition to facilitating an efficient implementation, also makes the worst-case stopping time of $\OptZs$ and $\OptZ$ apparent.

\begin{lemma}\label{optz-f1}
In $\OptZs$ (and thus also $\OptZ$), all decisions are made by time $f+1$ at the latest.%
\footnote{In all our protocols, a process can stop at the earlier of one round after deciding and time~$\tee+1$.}
\end{lemma}

It is interesting to compare $\OptZ$ with  efficient early-stopping consensus protocols \cite{CBS-uni,DRS,GGP,HalMoWa2001}. 
Let's say that \defemph{the sender set repeats at~$\node{i,m}$} in run~$r$ if~$i$ hears from the same set of processes in rounds~$m-1$ and~$m$. If this happens then, for every $\node{j,m-1}\notin r_i(m)$, we are guaranteed that  $(\node{j,m-2},\node{i,m-1})\notin r_i(m)$. Thus, all nodes at time $(m-1)$ are revealed to~$\node{i,m}$. Indeed, in a run in which~$f$ failures actually occur, the sender set will repeat for every correct process by time~$f+1$ at the latest. Efficient early stopping protocols typically decide when the sender set repeats. Indeed, the protocol $\hmwopt$ that was claimed by \cite{HalMoWa2001} to be unbeatable does so as well, with a slight optimization.  
Writing $\forall 1$ to stand for ``all initial values are~1'',  $\hmwopt$ is described as follows:

\vspace{\topsep}
\noindent
~~\underline{~{\bf Protocol}~$\hmwopt$~}
 (for an undecided process~$i$ at time~$m$)~\cite{HalMoWa2001}~:\\
\begin{tabular}{lll}
\quad\qquad  {\bf if} & 
$K_i\exZ$ 
&  {\bf then} $\decideiZ$\\
\quad\qquad  {\bf elseif} &   $K_i\forall 1$ ~or ~$m\ge 2$ and the sender set repeats at~$\node{i,m}$
&{\bf then} $\decideiO$
\end{tabular}
\vspace{\topsep}

$\OptZ$ and $\hmwopt$ differ only in the rule for deciding~1. But~$\OptZ$ strictly beats~$\hmwopt$, and sometimes by a wide margin. If $t=\Omega(n)$ then it can decide faster by a {\em ratio} of $\Omega(n)$. Indeed, we can show: 
\begin{lemma}
\label{lem:beats}
If $3\le\tee\le n-2$, then $\OptZ$ strictly dominates $\hmwopt$. Moreover, there exists an adversary
for which $\decideiO$ is performed after~3 rounds in $\OptZ$, and after $\tee+1$ rounds in $\hmwopt$. 
\end{lemma}

\subsection{Hidden Paths and Agreement}\label{hidden-paths}

It is instructive to examine the proof of \cref{lem:rev} (see Appendix~\ref{sec-proofs-consensus})
and consider when an active process~$i$ is undecided at~$\node{i,m}$  in~$\OptZ$. This occurs if both $\neg K_i\exZ$ and, in addition, for every $k=0,\ldots,m$ there is at least one node $\node{j_k,k}$ that is not revealed to~$\node{i,m}$. We call the sequence of nodes $\node{j_0,0},\ldots,\node{j_m,m}$ a \defemph{hidden path w.r.t.}~$\boldsymbol{\node{i,m}}$.
Such a hidden path implies that all processes $j_0,\ldots,j_m$ have crashed.
Roughly speaking, $\exZ$ could be relayed 
along such a hidden path without~$i$ knowing it (see \cref{hidden-path}).
\pgfmathsetmacro{\graphrowspace}{.8}%
\pgfmathsetmacro{\graphradius}{.163}%
\newcounter{aaa}%
\setcounter{aaa}{0}%
\newcounter{bbb}%
\setcounter{bbb}{0}%
\begin{figure}[ht]%
\centering%
\subfigure[\scriptsize All nodes seen (directly or indirectly) by $\node{i,3}$. The initial value is shown for all seen time-$0$ nodes. 
Notably, both $\lnot K_i\exists0$ and $\lnot K_i\lnot\exists0$ hold at time $m=~3$.]{%
\begin{tikzpicture}[font=\tiny,scale=0.69]
\node at (-1,4*\graphrowspace) {$i$};
\foreach \i in {0,1,2,3}
{
\node at (-1,{(3-\i)*\graphrowspace}) {$j_{\i}$};
}
\node at (-1,-\graphrowspace) {$m:$};
\node at (3,4*\graphrowspace) [above right] {$\node{i,3}$};
\foreach \m in {0, 1, 2, 3}
{
\draw[fill] (\m,4*\graphrowspace) circle (\graphradius);
\ifthenelse{\m = 0}{
\node[white] at (\m,4*\graphrowspace) {$1$};
}{}
\foreach \i in {0,1,2,3}
{
\ifthenelse{\m < \i}{
\draw[fill] (\m,{(3-\i)*\graphrowspace}) circle (\graphradius);
\ifthenelse{\m = 0}{
\node[white] at (\m,{(3-\i)*\graphrowspace}) {$1$};
}{}
}{}
}
\node at (\m,-\graphrowspace) {$\m$};
}
\end{tikzpicture}%
}\quad
\subfigure[][\scriptsize The state of each node, according to the information held by $\node{i,3}$: \newline
\mbox{\includegraphics[width=.7em]{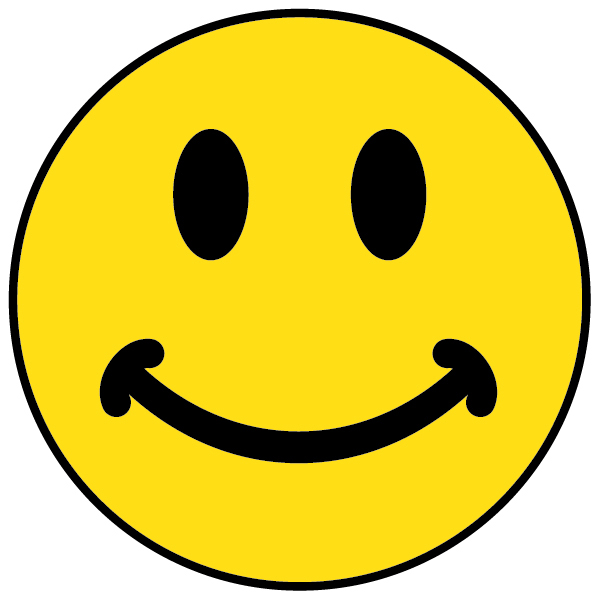}}=seen by all;
\mbox{\includegraphics[width=.7em]{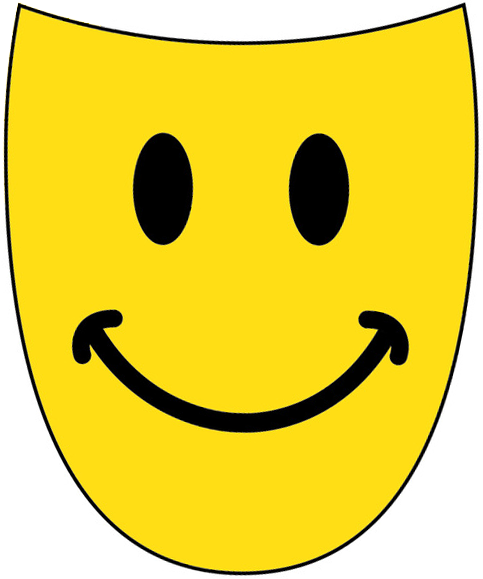}}=seen, may have crashed; 
\mbox{\includegraphics[width=.7em]{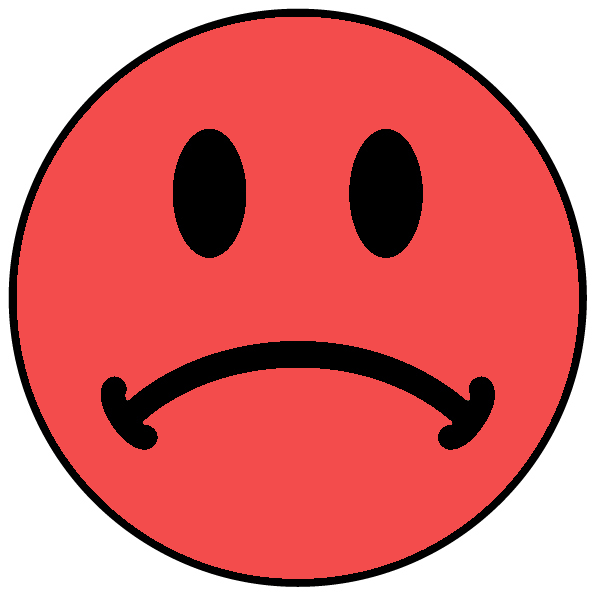}}=revealed, seen by none;
\mbox{\includegraphics[width=.7em]{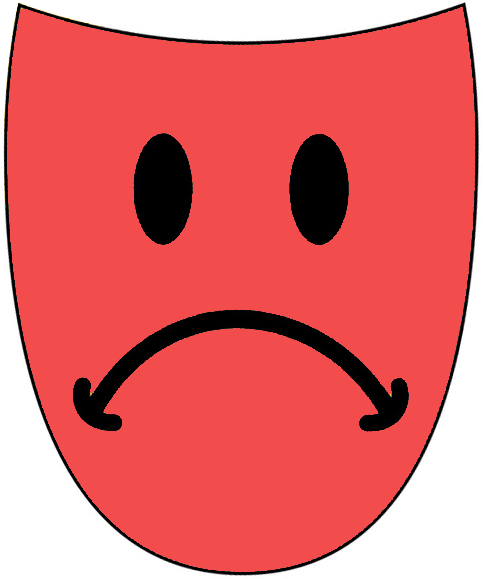}}=hidden: may have been seen by others.\ifthenelse{\theaaa=0}{\stepcounter{aaa}}{\footnotemark}]{%
\begin{tikzpicture}[font=\tiny,scale=0.69]
\node at (-1,4*\graphrowspace) {$i$};
\foreach \i in {0,1,2,3}
{
\node at (-1,{(3-\i)*\graphrowspace}) {$j_{\i}$};
}
\node at (-1,-\graphrowspace) {$m:$};
\node at (3,4*\graphrowspace) [above right] {$\node{i,3}$};
\foreach \m in {0, 1, 2, 3}
{
\node at (\m,4*\graphrowspace) {\includegraphics[width=1em]{smiley.jpg}};
\foreach \i in {0,1,2,3}
{
\ifthenelse{\m < \i}{
\ifnumcomp{\m}{=}{\i-1}{
\node at (\m,{(3-\i)*\graphrowspace}) {\includegraphics[width=1em]{smiley_mask.png}};
}{
\node at (\m,{(3-\i)*\graphrowspace}) {\includegraphics[width=1em]{smiley.jpg}};
}}{
\ifthenelse{\m > \i}{
\node at (\m,{(3-\i)*\graphrowspace}) {\includegraphics[width=1em]{frownie_red.png}};
}{
\ifthenelse{\m < 3 \OR \i < 3}{
\node at (\m,{(3-\i)*\graphrowspace}) {\includegraphics[width=1em]{frownie_mask_red.png}};
}{
\node at (\m,{(3-\i)*\graphrowspace}) {$?$};
}
}
}
}
\node at (\m,-\graphrowspace) {$\m$};
}
\end{tikzpicture}%
}\quad
\subfigure[][\scriptsize A run that is possible according to the information held by $\node{i,3}$;\ifthenelse{\thebbb=0}{\stepcounter{bbb}}{\footnotemark}\ in this run, $K_{j_3}\exists0$ holds at time $m=3$. Therefore, $\lnot K_i\notnz$ at time $m=3$. $\node{i,3}$ is therefore undecided in $\OptZ$.]{%
\begin{tikzpicture}[font=\tiny,scale=0.69]
\node at (-1,4*\graphrowspace) {$i$};
\foreach \i in {0,1,2,3}
{
\node at (-1,{(3-\i)*\graphrowspace}) {$j_{\i}$};
}
\node at (-1,-\graphrowspace) {$m:$};
\node at (3,4*\graphrowspace) [above right] {$\node{i,3}$};
\foreach \m in {0, 1, 2, 3}
{
\draw[fill] (\m,4*\graphrowspace) circle (\graphradius);
\ifthenelse{\m = 0}{
\node[white] at (\m,4*\graphrowspace) {$1$};
}{}
\foreach \i in {0,1,2,3}
{
\ifthenelse{\m < \i}{
\draw[fill] (\m,{(3-\i)*\graphrowspace}) circle (\graphradius);
\ifthenelse{\m = 0}{
\node[white] at (\m,{(3-\i)*\graphrowspace}) {$1$};
}{}
}{
\ifthenelse{\m = \i}{
\draw[fill,color=gray] (\m,{(3-\i)*\graphrowspace}) circle (\graphradius);
\ifthenelse{\m = 0}{
\node[white] at (\m,{(3-\i)*\graphrowspace}) {$0$};
}{}
\ifthenelse{\m <3 0}{
\draw[->,>=latex,shorten >=8pt,shorten <=8pt,color=gray] (\m,{(3-\i)*\graphrowspace}) -- (\m+1,{(3-\i-1)*\graphrowspace});
\node at (\m+.45,{(3-\i-.45)*\graphrowspace}) {\includegraphics[width=.5em]{smiley_mask.png}};
}{}
}{}
}
}
\node at (\m,-\graphrowspace) {$\m$};
}
\end{tikzpicture}%
}%
\caption{\small A hidden path $\node{j_0,0},\ldots,\node{j_3,3}$ w.r.t.\ $\node{i,3}$ implies $\lnot K_i\notnz$ at $3$.}%
\label{hidden-path}%
\end{figure}%
\addtocounter{footnote}{-1}%
\footnotetext{\label{double-mask}For simplicity, in this example every node seen by $\node{i,3}$ is also seen by all other nodes in the view of $\node{i,3}$. In other words, there exists no node $\node{j,m'}$ that is in state \includegraphics[width=.7em]{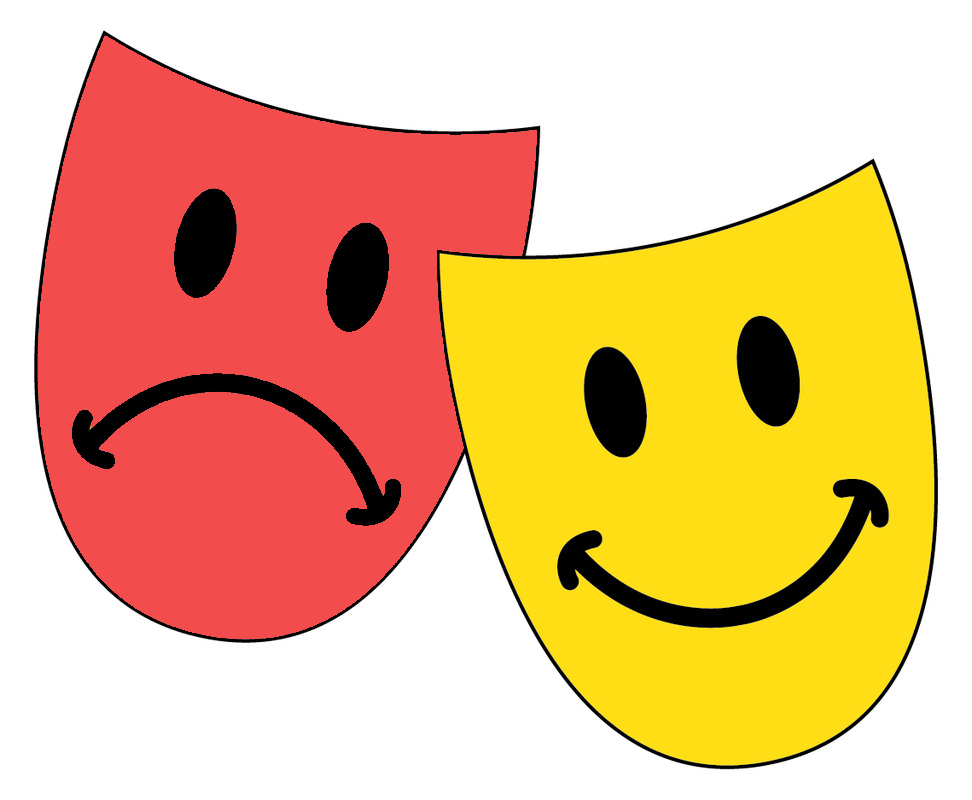} according to the information held by $\node{i,3}$, i.e.\ both $\node{j,m'}$ is seen by $\node{i,3}$, and $i$ has indirectly learnt by time $3$ that $j$ has in fact crashed at $m'$.}%
\addtocounter{footnote}{1}%
\footnotetext{In this run, the state of both $\node{j_0,0}$ and $\node{j_1,1}$, according to the information held by $\node{j_3,3}$, is \includegraphics[width=.7em]{double_mask_red.png}, as defined in \cref{double-mask}.}%
More formally,  its existence means that there is a run, indistinguishable at $\node{i,m}$ from the current one, in which 
$\val_{j_0}=0$ and this fact is sent from each $j_k$ to $j_{k+1}$ in every round $k+1\le m$. In that run process~$j_m$ is active at time~$m$  and $K_{j_m}\exZ$, and that is why $K_i\notnz$ {\em does not} hold. 
Hidden paths are implicit in many lower bound proofs for consensus in the crash failure model~\cite{DRS,DM}, but they have never before been captured formally. Clearly, hidden paths can relay more than just the existence of a value of~0. In a protocol in which some view can prove that the state is univalent in the sense of Fischer, Lynch and Paterson \cite{FLP}, a hidden path from a potentially pivotal state can keep processes from deciding on the complement value. 
Our analysis in the remainder of the paper provides additional cases in which unbeatable consensus is obtained when hidden paths can be ruled out. 

\subsection{Majority Consensus}

Can we obtain other unbeatable consensus protocols? Clearly, the symmetric protocol $\OptO$, obtained from~$\OptZ$ by reversing the roles of~0 and~1, is unbeatable and neither dominates, nor is dominated by, $\OptZ$. 
Of course, $\OptZ$ and~$\OptO$ are extremely biased, each deciding on its favourite value if at all possible, even if it appears as the initial value of a single process.
One may argue that it is natural, and may be preferable in many applications, to seek a more balanced solution, in which minority values are not favoured. Fix~$n>0$ and define the fact ``$\Maj=0$'' to be true if at least $n/2$ initial values are~0, while ``$\Maj=1$'' is true if strictly more than $n/2$ values are~1. Finally, relative to a node $\node{i,m}$, we define $\Majvals{i,m}\eqdef 0$ if at least half of the processes whose initial value is known to~$i$ at time~$m$ have initial value~$0$; $\Majvals{i,m}\eqdef 1$ otherwise. 
Consider the following protocol:

\vspace{\topsep}
\noindent
~~\underline{~{\bf Protocol}~$\OptMaj$~}
 (for an undecided process~$i$ at time~$m$):\\
\begin{tabular}{lll}
\qquad\qquad  {\bf if} & $K_i(\Maj=0)$  &  {\bf then} $\decideiZ$\\
\qquad\qquad  {\bf elseif} & $K_i(\Maj=1)$ &{\bf then} $\decideiO$\\
\qquad\qquad  {\bf elseif} &
some time $k\le m$ is revealed to~$\node{i,m}$  
& {\bf then} $\decidei(\Majvals{i,m})$.
\end{tabular}
\vspace{\topsep}

We note that whether $K_i(\Maj=0)$ (resp.\ $K_i(\Maj=1)$) holds can be checked efficiently: it holds exactly if $i$ has seen at least (resp.\ strictly more than) $n/2$ time-$0$ nodes with initial value $0$ (resp.\ $1$).

\begin{theorem}
\label{thm:optmaj}
If $\tee>0$, then
$\OptMaj$ is an unbeatable consensus protocol. In particular, in a run in which~\mbox{$f\le \tee$} failures actually occur, all decisions are performed by time~$f+1$, at the latest. 
\end{theorem}

The proof of \cref{thm:optmaj} formalizes the following idea. Suppose that~$i$ sees fewer than a full majority of either value at $\node{i,m}$ and has a hidden path. Then~$i$ considers it possible that the node~$\node{j_1,1}$ in the hidden path may have seen either a full majority of~0's or a full majority of~1's, 
and this information may reach an active node $\node{j_m,m}$. Decision is thus impossible in this case, and decisions are made when no hidden path w.r.t.~$\node{i,m}$ is possible. 
Thus, $\OptMaj$ is an unbeatable consensus protocol that
satisfies an additional ``fairness'' property: 

\begin{itemize}
\item[]{\bf Majority Validity:}\quad For $\veee\in\{0,1\}$, if more than half of the processes are both correct and have initial value~$\veee$, then 
no process decides~$\bar\veee$ in~$r$.
\end{itemize}

\section{Unbeatable Uniform Consensus}
\label{sec:uniform}

It is often of interest to consider {\em uniform} consensus \cite{CBS-uni,Dutta-uni,H86,KR-uni,Raynal04-uni,WTC-uni}  in which 
we replace the \Agreement\ condition of consensus by: 

\begin{itemize}
\item[]{\bf Uniform~Agreement:}\quad The processes that decide in a given run must all decide on the same value. 
\end{itemize}

This forces correct processes and faulty ones to act in a consistent manner.
Requiring uniformity makes sense only in a setting where failures are benign, and all processes that decide do so according to the protocol. 
Uniformity may be desirable when elements outside the system can observe decisions, as in distributed databases when decisions correspond to commitments to values.

Under crash failures, a process generally does not know whether or not it is correct. Indeed, so long as it has not seen~$\tee$ failures, the process may (for all it knows) crash in the future. As a result, while~$K_i\exists{0}$
is a necessary condition for $\decideiZ$ as before, it cannot be a {\em sufficient} condition for decision in any uniform consensus protocol. This is because a process starting with~0 immediately decides~0 with this rule, and may immediately crash. If all other processes have initial value~1, all other decisions can only be on~1. 
Of course, $K_i\exists{0}$ is still a necessary condition  for deciding~0, but it is {\em not} sufficient. Denote by $\cv$ the fact ``some \defemph{correct} process knows $\existsv$''. We show the following:

\begin{lemma}
\label{lem:correct-uni}
${K_i\cv}$ is a necessary condition  for~$i$ deciding~$\veee$ in any protocol solving Uniform Consensus.
\end{lemma}

There is a direct way to test whether $K_i\cv$ holds, based on $r_i(m)$: 

\begin{lemma}
\label{lem:u-know}
Let $r\in R_P=R(P,\gammacr)$  and assume that $i$ knows of $\defemph{d}$ failures at $(r,m)$. 
Then
$(R_P,r,m)\sat K_i\cv$ ~iff~ at least one of ~~(a)~~$m\!>\!0$
and $(R_P,r,m\!-\!1)\sat K_i\existsv$, ~or
(b)~~$(R_P,r,m)\sat K_i(\mbox{$K_j\existsv$ ~held at time~$m\!-\!1$})$ ~holds for at least $(\tee\!-\!\defemph{d})$ distinct processes~$j$.
\end{lemma}

By \cref{lem:knowing0}, at time~$\tee+1$ the conditions $K_i\existsv$ and $K_i\cv$ are equivalent.
As in the case of consensus, we note that if $K_i\exists0$ (equivalently, $K_i\cz$) does not hold at time $\tee+1$, then it  never will.
We thus phrase the following \emph{beatable} algorithm, analogous to $\Pz$, for Uniform Consensus; in this protocol, $K_i\cz$ (the necessary condition for
deciding~$0$ in uniform consensus) replaces $K_i\exists0$ (the necessary condition in consensus) as the decision rule for~$0$. The decision rule for~$1$ remains the same.
Note that $K_i\cz$ can be efficiently checked, by applying the test of~\cref{lem:u-know}.

\vspace{\topsep}
\noindent
~~\underline{~{\bf Protocol}~$\UPz$~}
 (for an undecided process $i$ at time $m$):\\
\begin{tabular}{lll}
\qquad\qquad{\bf if} & $K_i\cz$ & {\bf then} $\decideiZ$ \\
\qquad\qquad{\bf elseif} & $m=\tee+1$ & {\bf then} $\decideiO$.
\end{tabular}
\vspace{\topsep}

Following a similar line of reasoning to that leading to $\OptZ$, 
we 
obtain an unbeatable uniform consensus protocol:

\vspace{\topsep}
\noindent
~~\underline{~{\bf Protocol}~$\UOptZ$~}
 (for an undecided process~$i$ at time~$m$):\\
\begin{tabular}{lll}
\qquad\qquad  {\bf if} & $K_i\cz$  &  {\bf then} $\decideiZ$\\
\qquad\qquad{\bf elseif}
& $\neg K_i\exists{0}$ 
and some time $k\le m$ is revealed to~$\node{i,m}$  
&{\bf then} $\decideiO$.
\end{tabular}
\vspace{\topsep}

Recall that whether $K_i\cz$ holds can be checked efficiently via the characterization in \cref{lem:u-know}.

\begin{theorem}
\label{thm:u-opt}
$\UOptZ$ is an unbeatable \defemph{uniform} consensus protocol in which all decisions are made by time $f+2$ at the latest, and if $f \ge \tee-1$, then all decisions are made by time $f+1$ at the latest.
 \end{theorem}

Hidden paths again play a central role. Indeed,
as in the construction of $\OptZ$ from $\Pz$, the construction of $\UOptZ$ from $\UPz$ involves some decisions on $1$ being moved earlier in time, by means of the last condition, checking the absence of a hidden path. (Decisions on $0$ cannot be moved any earlier, as they are taken as soon as the necessary condition for deciding $0$ holds.)
Observe that the need to obtain $K_i\cv$ rather than $K_i\exv$ concisely captures the essential distinction between uniform consensus and nonuniform consensus. The fact that the same condition --- the existence of a hidden path --- keeps a process~$i$  from knowing that no active~$j$ can know $K_j\cv$, as well as keeping~$i$ from knowing that no~$j$ knows $K_j\exv$, explains why the bounds for both problems, and their typical solutions,  are similar. 

Proving the unbeatability of $\UOptZ$ is  more  challenging than proving it for $\OptZ$. Intuitively, this is because 
gaining that an initial value of $0$ that is known by a nonfaulty process does not imply that some process has already decided on $0$.
As a result, the possibility of dominating $\UOptZ$ by switching~0 decisions to~1 decisions needs to be explicitly rejected. This is done by employing  
reachability arguments essentially establishing the existence of the continual common knowledge conditions of~\cite{HalMoWa2001}.

The fastest early-stopping protocol for uniform consensus in the literature, opt-EDAUC of~\cite{CBS-uni}
(a similar algorithm is in~\cite{Dutta-uni}), also stops in $\min(f+2,\tee+1)$ rounds at the latest. Similarly to \cref{lem:beats}, 
not only does 
$\UOptZ$ strictly dominate opt-EDAUC, but furthermore, there are adversaries against which $\UOptZ$ decides in~1 round, while opt-EDAUC decides in $\tee+1$ rounds:

\begin{lemma}\label{lem:ubeats}
If $2\le \tee\le n-2$, then $\UOptZ$ strictly dominates the opt-EDAUC protocol of \cite{CBS-uni}. Moreover, there exists an adversary
for which $\decideiO$ is performed after~1 round in $\UOptZ$, and after $\tee+1$ rounds in opt-EDAUC. 
\end{lemma}

\section{Discussion}
\label{sec:discussion}

It is possible to consider variations on the notion of unbeatability.
One could, for example, compare runs in terms of the time at which the last correct process decides.
We call the corresponding notion \defemph{last-decider unbeatability}.%
\footnote{This notion was suggested to us by Michael Schapira; we thank him for the insight.}  
This neither implies, nor is implied by, the notion of unbeatability studied so far in this paper. 
None of the consensus protocols in the literature is last-decider unbeatable.
In fact, all of our protocols are also last-decider unbeatable:

\begin{theorem}
\label{thm:last-decider}
The protocols $\OptZ$ and $\OptMaj$ are also last-decider unbeatable for consensus, 
while~$\UOptZ$ is
last-decider unbeatable for uniform consensus.
\end{theorem}

We note that \cref{lem:beats,lem:ubeats} show that our protocols beat the previously-known best ones by a large margin w.r.t.\ last-decider unbeatability as well.

Unbeatability is a natural optimality criterion for distributed protocols. 
It formalizes the intuition that a given protocol cannot be strictly improved upon, which is significantly stronger than saying that it is worst-case optimal, or even early stopping. All of the protocols that we have presented have a very concise and intuitive description, and are efficiently implementable;
thus, unbeatability is attainable at a modest price. 
Crucially, our unbeatable protocols can decide much faster than previously known solutions to the same problems.

\section*{Acknowledgements}
Armando Casta\~{n}eda was supported in part by an Aly Kaufman Fellowship at the Technion.
Yannai Gonczarowski was supported in part by ISF grant 230/10, by the Google Inter-university center
for Electronic Markets and Auctions, by
the European Research Council under the European Community's Seventh Framework
Programme (FP7/2007-2013) / ERC grant agreement no.\ [249159] and
by an Adams Fellowship of the Israeli Academy of Sciences and Humanities.
Yoram Moses is the Israel Pollak Academic chair at the Technion; his work was supported in part by ISF grant 1520/11.

\bibliographystyle{abbrv}
\bibliography{z}

\appendix

\section{Proofs}

\subsection{Consensus}
\label{sec-proofs-consensus}

\begin{proof}[Proof of \cref{lem:necessary}]

This proof uses notation introduced in \cref{sec:model}.
Let $P$ be a consensus protocol and let $R_P=R(P,\gammacr)$. Let $\valv\in\Vals$, let $r\in R_P$ and let $\node{i,m}$ be a node s.t.\ $i$ decides on $\valv$ at time $m$ in $r$.

We commence by proving (a).
Assume for contradiction that no process has initial value $\valv$ in $r$.
By definition of $\gammacr$,
there exists a run $r'$ of $P$, s.t.~~\emph{1)} $r'_i(m)\!=\!r_i(m)$,\ \ \emph{2)} $i$ does not fail in $r'$, and\ \ \emph{3)} The initial values in $r'$ are the same as in $r$.
As $r'_i(m)=r_i(m)$, we have that $i$ decides on $\valv$ at time $m$ in $r'$ as well. As the initial values in $r'$ are the same as in $r$, we have that no process has initial value $\valv$ in $r'$. As $i$ does not fail in $r'$, we therefore have that \Validity\ does not hold regarding the decision of $i$ in $r'$ --- a contradiction.

We move on to proving (b).
Assume for contradiction that some process $j$ decides $\bar\valv$ at some time $m'\le m$ in $r$, and that $j$ is active at $m$ in $r$.
Once again by definition of $\gammacr$,
there exists a run $r'$ of $P$, s.t.\ \ \emph{1)} $r'_i(m)\!=\!r_i(m)$,\ \ \emph{2)} $r'_j(m')\!=\!r_j(m')$, and\ \ \emph{3)} neither $i$ nor $j$ fail in $r'$.
As $r'_i(m)=r_i(m)$, we have that $i$ decides on $\valv$ at time $m$ in $r'$ as well; as $r'_j(m')=r_j(m')$, we have that $j$ decide on $\bar\valv$ at time $m'$ in $r'$ as well. As neither $i$ not $j$ fail in $r'$, we therefore have that \Agreement\ does not hold in $r'$ --- a contradiction.
\end{proof}

\begin{proof}[Proof of \cref{lem:know-exists}]

Directly from \cref{lem:necessary} and \cref{thm:knowprec}.
\end{proof}

While \cref{lem:knowing0} is given and proved in \cite{DM}, for completeness we reprove it here using the notation and machinery of this paper; this proof is assisted by \cref{chain,0-chain}.

\begin{definition}\label{chain}
Let $P$ be a protocol in~$\gammacr$ and let $r\in R_P=R(P,\gammacr)$. Let $\valv\in\Vals$ and let $\node{i,m}$ be a node.
We say that \defemph{there is a $\valv$-chain for $\node{i,m}$ in the run~$r$} if, 
for some $d\le m$, there is a sequence $j_0,j_1,\ldots, j_d=i$ of distinct processes, such that 
$\val_{j_0}=\valv$ and for all $1\le k\le d$, the process~$j_k$ receives a message from~$j_{k-1}$ at time~$k$ in~$r$.
\end{definition}

\begin{lemma}\label{0-chain}
Let $P$ be a \fip\ in~$\gammacr$ and let $r\in R_P=R(P,\gammacr)$. Then for every processes $i$
and time~$m\ge 0$, it is the case that
$(R_P,r,m)\sat K_i\exZ$ iff there is a $0$-chain for $\node{i,m}$ in~$r$.
\end{lemma}

\begin{proof}
For the first direction, assume that there is a $0$-chain $j_0,\ldots,j_d=i$ for $\node{i,m}$ in $r$. It is easy to show by induction that $K_{j_k}\exists0$ at $k$ in $r$ for every $k$; therefore, $K_i\exists0$ at $d$ in $r$, and since $P$ is a fip, $K_i\exists0$ at $m$ in $r$, as required.
We prove the second direction for all $i$ by induction on $m$.

Base ($m=0$): Since process $i$ at time $0$ knows no initial value but its own, we have that $\val_i=0$ and so $i$ (with $d=0$) is a $0$-chain as required.

Inductive step ($m>0$): In a fip, $K_i\exists0$ at $m$ implies that either $K_i\exists0$ at $m-1$ or $K_j\exists0$ at $m-1$ for some $j \neq i$ that successfully sends a message at time $m-1$ to $j$. If $K_i\exists0$ at $m-1$, then by the induction hypothesis there exists a $0$-chain for $\node{i,m-1}$ in $r$, and by definition this is also a $0$-chain for $\node{i,m}$ in $r$. It remains to consider the case in which $K_i\exists0$ does not hold at $m-1$; therefore, $K_j\exists0$ at $m-1$ for some $j$ that successfully sends a message at time $m-1$ to $j$. By the induction hypothesis, there exists a $0$-chain $j_0,\ldots,j_d=j$ for $\node{j,m-1}$. We first claim that $i$ does not appear in that chain; indeed, if $j_{d'}=i$ for some $d'<d$, then by definition
$j_0,\ldots,j_{d'}$ would be a $0$-chain for $\node{i,m-1}$, and by the previous direction we would have $K_i\exists0$ at $m-1$ in $r$. We now claim that $d=m-1$; indeed, if $d<m-1$, then $j_0,\ldots,j_d$ would be a $0$-chain for $\node{j,d}$, and so we would have $K_j\exists0$ at $d<m-1$. As $j$ is active at all times earlier than $m-1$, we would have that $\node{j,d}$ successfully sends a message to $i$, and so $K_i\exists0$ at $d+1\le m-1$; as $P$ is a fip, we would therefore have that $K_i\exists0$ at $m-1$ --- a contradiction. As $i$ does not appear in $j_0,\ldots,j_d$, and as $d=m-1$, by definition $j_0,\ldots,d_j,i$ is a $0$-chain for $i$, as required.
\end{proof}

\begin{proof}[Proof of \cref{lem:knowing0}]

Assume that $(R_P,r,\tee+1)\sat K_i\exv$. By \cref{0-chain}, there exists a $0$-chain $j_0,\ldots,j_d$ for $\node{i,\tee+1}$.
If $j$ appears in $j_0,\ldots,j_d$, then by \cref{0-chain} we are done; assume, therefore, that $j$ does not appear in $j_0,\ldots,j_d$.
If $d<\tee+1$, then since $i$ successfully sends all messages at times earlier than $\tee+1$, we have that $j_0,\ldots,j_d,j$ is a $0$-chain for $\node{j,\tee+1}$; therefore, by \cref{0-chain}, $K_j\exv$ at $\tee+1$, as required. Otherwise, $d=\tee+1$, and so, as $j_0,\ldots,j_{d-1}$ are $\tee+1$ distinct processes, there exists $0\le d'\le d-1$ s.t.\ $j_{d'}$ is nonfaulty throughout $r$. Therefore, $j_0,\ldots,j_{d'},j$ is a $0$-chain for $\node{j,\tee+1}$, as required.
\end{proof}

\begin{proof}[Proof of Lemma~\ref{lem:decide-when-0}]

Assume that $Q\dom\Pz$ solves consensus; w.l.o.g., $Q$ is a fip as well.
We prove the claim for all processes~$i$ and adversaries~$\alpha$, by induction on the time~$m$ at which 
$K_i\exists 0$ first holds in~$Q[\alpha]$ (and, equally, in $\Pz[\alpha]$). 

Base ($m=0$):\quad
As $i$ decides $0$ at time~$0$ in~$\Pz[\alpha]$, by \cref{lem:know-exists} we have $K_i\exists0$ at time $0$ in $\Pz[\alpha]$ (and so also in $Q[\alpha]$).
Since process~$i$ at time~$0$ knows no initial value but its own, it follows that $i$ is assigned an initial value of~$0$ by $\alpha$. Hence, $K_i\exists{1}$ does \emph{not} hold at~$0$.
By \cref{lem:know-exists}, $i$ therefore does not decide~1 at time~$0$ in~$Q[\alpha]$. Since $i$ decides at time~0 in~$\Pz[\alpha]$, it must decide at time~$0$ in~$Q[\alpha]$ as well, and so decides~0, as required.

Inductive step ($m>0$): Assume that the claim holds for all times $<m$.
Recall that $m$ is the first time at which
$K_i\exists{0}$ holds.
In a fip, this can only happen if 
$K_i\exists{0}$ does not hold at time $m' < m$ and $i$ 
receives at time $m$ a message with a~0 from some process~$j$ that is active at time~$m-1$.
Thus, $K_j \exists 0$ holds at time $m-1$, and by the induction hypothesis,
$j$ decides $0$ when $K_j\exists{0}$ first holds in $Q[\alpha]$ --- denote this time by $m'$; as $K_j \exists 0$ holds at time $m-1$, we have $m'\le m-1$.
Observe that in~$\gammacr$, if $i$ receives a message from~$j$ in round~$m$, then~$i$ cannot know that~$j$ is faulty at time~$m$; more precisely,
denoting by $\beta$ the adversary that never crashes~$i$ nor $j$ at all, and that otherwise agrees with $\alpha$ (this is a legal adversary, as is specifies no more than $t$ crash failures),
we have in the run $r'=Q[\beta]$ that\ \ \emph{1)} $r'_i(m)\!=\!r_i(m)$,\ \ \emph{2)} $r'_j(m')\!=\!r_j(m')$, and\ \ \emph{3)} neither $i$ nor $j$ fail.
Since $Q$ satisfies \Agreement, $i$ cannot decide~1 during $Q[\beta]$, and therefore cannot decide~1 at or before time~$m$ during $Q[\alpha]$.
Moreover,
by \cref{lem:know-exists}, $K_i\exists{0}$ is a precondition for process~$i$ deciding~0, and so $i$ cannot decide~$0$ before time $m$ during $Q[\alpha]$. 
Since~$Q$ dominates~$\Pz$, we have that~$i$ must decide by time~$m$ in~$Q[\alpha]$, and therefore it decides~0 at~$m$ in~$Q[\alpha]$.
\end{proof}

\begin{proof}[Proof of \cref{lem:decide-when-notnz}]

$\Longrightarrow$: 
Assume that $(R_P,r,m)\not\sat K_i\notnz$. Therefore, by definition of $K_i$, there exists a run $r'\in R_P$ s.t.\ \emph{1)} $r'_i(m)\!=\!r_i(m)$, and\ \ \emph{2)} $(R_P,r',m)\not\sat \notnz$.
As $(R_P,r',m)\not\sat \notnz$, there exists a process $j$ s.t.\ $K_j\exists0$ holds at $m$ in $r'$ (and $j$ is active at $m$ in $r'$). By definition, $K_j\exists0$ first holds at or before time $m$ in $r'$, and so $j$ decides~$0$ before or at time $m$ in $r'$; therefore, $(R_P,r',m)\not\sat \noactnO$. As $r'_i(m)=r_i(m)$, we therefore have $(R_P,r,m)\not\sat K_i\noactnO$, as required.

$\Longleftarrow$: 
We will show that $(R_P,r,m)\sat\notnz$ implies $(R_P,r,m)\sat\noactnO$; by definition of knowledge, it will then follow that $(R_P,r,m)\sat K_i\notnz$ implies $(R_P,r,m)\sat K_i\noactnO$.
Assume, therefore, that $(R_P,r,m)\sat\notnz$, and let $j$ be a process that is active at time $m$ in $r$. As $\notnz$ at $m$ in $r$, we have that $K_j\exists0$ does not hold at $m$ in $r$.
As $P$ is a fip, we have that neither does $K_j\exists0$ hold at any time prior to $m$ in $r$. By definition, therefore $j$ does not decide $0$ before or at $m$ in $r$, as required.
\end{proof}

\begin{lemma}\label{optz-t1}
Let $P$ be a \fip\ in~$\gammacr$ and let $r\in R_P=R(P,\gammacr)$. Let $i$ be a process.
If $(R_P,r,\tee+1)\not\sat K_i\exists0$, then $(R_P,r,\tee+1)\sat K_i\notnz$.
\end{lemma}

\begin{proof}
By \cref{lem:knowing0}, we have that $\lnot K_i\exists0$ at time $\tee+1$ implies $\notnz$ at that time; by definition of knowledge, we therefore have that
$K_i(\lnot K_i\exists0\wedge m=\tee+1)$ implies $K_i\notnz$. As both the clock $m$ and the value of $\tee$ are common knowledge,
we therefore have that $K_i(\lnot K_i\exists0)$ at time $\tee+1$ implies $K_i\notnz$ at that time. Finally, by the definition of knowledge we have that $K_i(\lnot K_i\exists0)$ holds iff $\lnot K_i\exists0$ holds, and the proof is complete.
\end{proof}

\begin{theorem}\label{solve}
$\OptZ$ solves consensus in $\gammacr$.
\end{theorem}

\begin{proof}
In some run $r$ of $\OptZ$, let $i$ be a nonfaulty process.

\Decision:
By definition of $\OptZ$, for any process that is active at time $\tee+1$, if $i$ has not decided $0$ by that time, we have $\lnot K_i\exists0$ at that time. Therefore, by \cref{optz-t1}, we have that $K_i\notnz$ at that time and so $i$ decides upon $1$ if it is undecided. Therefore, all processes that are active at time $\tee+1$, and in particular all nonfaulty processes, decide by that time at the latest, and in particular decide at some point throughout the run, as required.

Henceforth, let $m$ be the decision time of $i$
and let $\valv$ be the value upon which $i$ decides.

\Validity: If $\valv=0$, then $K_i\exists0$ at $m$; thus, $\exists0$ as required. Otherwise, $K_i\exists0$ does not hold at $m$; therefore, the initial value of $i$ is $1$, and so $\exists1$ as required.

\Agreement:
It is enough to show that if $\valv=1$, then no correct process ever decides $0$ in the current run.
Indeed, if any nonfaulty process $j$ decided $0$ at some time $m'<m$, then $i$ would have received a message with a $0$ from $j$ at $m'+1\le m$, and so we would have $K_i\exists0$ at $m$. To complete the proof, it is enough to show that no process decides $0$ at any time $m'\ge m$; this follows by an easy inductive argument, using the fact that $\notnz$ at any time $m''$ implies $\notnz$ at $m''+1$.
\end{proof}

\begin{proof}[Proof of \cref{thm:optz}]

Correctness is shown in \cref{solve}.
We thus have to show that for every protocol consensus protocol $Q\dom\OptZ$, we also have $\OptZ\dom Q$. Let, therefore, $Q$ be a consensus protocol s.t.\ $Q\dom\OptZ$; w.l.o.g., $Q$ is a fip.

We first claim that $\OptZ\dom\Pz$. Indeed, whenever $\Pz$ decides upon $0$, so does $\OptZ$; let therefore $i$ be a process deciding upon $1$ in $\Pz$; by definition of $\Pz$, this decision is made at time $m=\tee+1$, and furthermore, $\lnot K_i\exists0$ at that time. By \cref{optz-t1}, we therefore have that $K_i\notnz$ at time, and so, $i$ decides upon $1$ in $\OptZ$ at that time if it has not already decided.

By transitivity of domination, we thus have that $Q\dom\Pz$. By \cref{lem:decide-when-0}, we therefore have that $\decideiZ$  is performed  in~$Q$ exactly when $K_i\exists{0}$ first holds; therefore, no decision on $0$ is made in $Q$ before $\OptZ$. Moreover, by \cref{lem:know-exists,lem:decide-when-notnz}, we therefore have that $K_i\notnz$ is a necessary conditions for $\decidei(1)$ in~$R_Q=R(Q,\gammacr)$. Therefore, no decision on $1$ is made in $Q$ before $\OptZ$. Therefore $\OptZ\dom Q$, as required, and the proof is complete.
 \end{proof}

\begin{proof}[Proof of \cref{lem:rev}]

We first claim that $(R_P,r,m)\sat\notnz$ iff for every $0\le k\le m$, there exists a process $j_k$ s.t.\ $K_{j_k}\exists0$  at time $k$ in $r$ --- we call such $j_0,\ldots,j_m$ a \defemph{0-path} for time $m$ in $r$; the proof is similar to (and simpler than) that of \cref{0-chain} and is left to the reader.

Assume first that some time $k\le m$ is revealed to~$\node{i,m}$ in~$r$.
As $(R_P,r,m)\not\sat K_i\exists0$, we thus have that no time-$k$ node $j$ satisfies $K_j\exists0$; therefore, no $0$-path exists for time $m$ in $r$, and so $(R_P,r,m)\sat\notnz$. Note that by definition of knowledge, time $k$ is revealed to $\node{i,m}$ in $r$ iff $(R_P,r,m)\sat K_i(\mbox{time $k$ is revealed to $\node{i,m}$ in the current run})$. Therefore, we have that time $k$ being revealed to $\node{i,m}$ implies not only $(R_P,r,m)\sat\notnz$, but also $(R_P,r,m)\sat K_i\notnz$, as required.

Assume now that no time $k\le m$ is revealed to $\node{i,m}$, i.e.\ that for every $k\le m$, there exists a time-$k$ node $\node{j_k,k}$ that is not revealed to $\node{i,m}$ in $r$ --- in \cref{hidden-paths}, we call such $j_0,\ldots,j_m$ a \defemph{hidden path} w.r.t.\ $\node{i,m}$ in $r$. We construct a run $r'\in R_P$ s.t.\ $r'_i(m)=r_i(m)$, in which $j_0,\ldots,j_m$ constitutes a $0$-path for time $m$ --- see \cref{hidden-path} in \cref{hidden-paths}. The adversary in $r'$ meets the following conditions, and otherwise coincides with that of $r$:
\begin{itemize}
\item
The initial value of $j_0$ is~$0$.
\item
For every $k<m$, the node $\node{j_k,k}$ crashes, successfully sending a message solely to $\node{j_{k+1},k+1}$.
\item
$\node{j_m,m}$ is nonfaulty.
\end{itemize}
It is straightforward to verify that $r'_i(m)=r_i(m)$, that no more crashes occur in $r'$ than in $r$, and that $j_0,\ldots,j_m$ indeed is a $0$-path for time $m$ in $r$.
As $(R_P,r',m)\not\sat\notnz$, and as $r'_i(m)=r_i(m)$, we therefore have that $(R_P,r,m)\not\sat K_i\notnz$, as required.
\end{proof}

\begin{proof}[Proof of \cref{optz-f1}]

Let $i$ be an undecided node at time $m$ in $\OptZs$; it is enough to show that $m\le f$.
As $i$ is undecided, by definition of $\OptZs$, for every $0\le k<m$, there exists a node process $j_k$ s.t.\ $\node{j_k,k}$ is not revealed to $\node{i,m}$. We first note that all of the nodes $j_k$ are faulty; indeed, as $\node{j_k,k}$ is not revealed to $\node{i,m}$, and as $k<m$, we have that $\node{i,k+1}$ receives no message from $\node{j_k,k}$. We further note that all $j_k$ are distinct; indeed, for every $k<k'<m$, we have (once again since $\node{j_k,k}$ is not revealed to $\node{i,m}$) that $\bigl(\node{j_k,k'-1},\node{i,k'}\bigr)\notin r_i(m)$ while $\node{i,k'}\in r_i(m)$, and so by definition $\node{j_k,k'}$ is revealed to $\node{i,m}$.
We conclude that $j_0,\ldots,j_{m-1}$ are $m$ distinct faulty nodes, and so $m\le f$ and the proof is complete.
\end{proof}

\begin{proof}[Proof of \cref{lem:beats}]

First notice that~$\OptZ$ dominates~$\hmwopt$, since $K_i\forall 1$ is true iff time $0$ is revealed to~$i$, and if $i$'s sender set repeats in round~$m$, then time~$m-1$ is revealed to~$i$ at time~$m$. Hence, for every adversary, processes decide in~$\OptZ$ at least as soon as they do in~$\hmwopt$. We now show an adversary for which the decisions are made strictly earlier in~$\OptZ$ than in~$\hmwopt$; moreover, this adversary meets the conditions of the second clause of the \lcnamecref{lem:beats}.

Denote the processes by~$\Proc=\{1,2,\ldots,n\}$. Let~$\alpha$ be defined as follows. 
All initial values in~$\alpha$ are~$1$. 
In round~$1$, only process~$1$ fails, and it is silent: it crashes without sending any messages. In round~$2$ two processes crash---process~$2$ and process~$3$, with process~$2$ sending only to process~$n$, and process~$3$ sending to everyone except process~$n$. No process fails in round~$3$, and, in each of the rounds $m=4,\ldots,\tee$, process~$m$ crashes without sending any messages. 
Since precisely~$\tee$ processes fail in~$\alpha$ we have that $\alpha\in\Crash(\tee)$.  

Observe that in $\fip[\alpha]$ no correct process ever knows process~$1$'s initial value. In addition, for every correct process, the first round in which the sender set repeats is round~$\tee+1$. 
Indeed, every correct process other than~$n$ fails to hear from process~$m$ for the first time in round~$m$, for $m=1,\ldots,\tee$, while process~$n$ differs slightly, in that it fails to hear from process~3 in round~2 and from process~2 in round~3. Therefore, in $\hmwopt[\alpha]$ all correct processes decide~$1$ at time~$\tee+1$, since round~$\tee+1$ is the first one in which their sender set repeats; no process decides any earlier.
Now let us consider when a process~$i$ that is correct according to~$\alpha$ decides in~$\OptZ$.  
By definition, $i$ receives messages in round~3 from both $\node{n,2}$ and $\node{n-1,2}$.  Together, these contain the information about nodes $\node{2,1},\node{3,1},\ldots,\node{n,1}$. Moreover, node $\node{1,1}$ is revealed  to~$\node{i,3}$ as well (as being crashed --- $\frownie$), since  the edge $(\node{1,0},\node{i,1})$ is absent from~$i$'s view at $\node{i,3}$. It follows that time $1$ is revealed to $\node{i,3}$, and so~$i$ decides~1 at time~3, after~3 rounds, as claimed. Since $3<4\le \tee+1$, we have that when the adversary is~$\alpha$, decisions in~$\OptZ$ occur strictly earlier than in~$\hmwopt$, and we are done.  
\end{proof}

\subsection{Majority Consensus}
\label{sec-maj-proofs}
The proof of \cref{thm:optmaj} is based on two lemmas:
\begin{lemma}[Decision at time $1$]
\label{maj-decide-first-round}
Assume that  $n\!>\!2$ and $t\!>\!0$. Let $Q\dom\OptMaj$ solve Consensus and let $r\!=\!Q[\alpha]$ be a run of $Q$.
Let $i$ be a process and let $\veee$ be a value. If $K_i(\Maj\!=\!\veee)$ at $(r,1)$, then $Q$ makes~$i$ decide~$\veee$ before or at time $1$ in~$r$.
\end{lemma}

\begin{proof}
By definition of $\OptMaj$, $i$ decides in $\OptMaj[\alpha]$ by time $1$, since $K_i(\Maj\!=\!\veee)$ holds at $(\OptMaj[\alpha],1)$. As $Q\dom\OptMaj$, we thus have that~$i$ must decide upon some value in $r\!=\!Q[\alpha]$ before or at time $1$. Thus, it is enough to show that $i$ cannot decide $1\!-\!\veee$ up to time $1$ in $r$.

We prove the claim by induction on $n\!-\!|Z_i|$,
where $Z_i$ is defined to be the set of processes $k$ with initial value $\veee$, s.t.\ $\node{k,0}$ is seen by $\node{i,1}$.
As $K_i(\Maj\!=\!\veee)$ at $(r,1)$, we have $|Z_i|\ge \frac{n}{2}$ and so $2 \le |Z_i| \le n$.

Base: $|Z_i|=n$. In this case, all initial values are $\veee$, and so by \Validity~$i$ cannot decide $1\!-\!\veee$ in $r$.

Step: Let $2\le\ell<n$ and assume that the claim holds whenever $|Z_i|=\ell+1$. Assume that $|Z_i|=\ell$. As $|Z_i| \ge 2$, there exists $j \in Z_i \setminus \{i\}$. We reason by cases.

\begin{enumerate}[label=\Roman*.]
\item
If there exists a process $k$ s.t.\ $\node{k,0}$ is hidden from $\node{i,1}$, then there exists a run $r'$ of $Q$, s.t.~~\emph{1)} $r'_i(1)\!=\!r_i(1)$,\ \ \emph{2)}~neither $i$ nor $j$ fail in $r'$,\ \ \emph{3)} $k$ has initial value $0$ in $r'$, and\ \ \emph{4)} $Z_j = Z_i\!\cup\!\{k\}$ in $r'$. (Note that by definition, $Z_i$ has the same value in both $r$ and $r'$.)
By the induction hypothesis (switching the roles of $i$ and $j$), $j$ decides $\veee$ before or at time $1$ at $r'$, and therefore by \Agreement, $i$ cannot decide $1\!-\!\veee$ in $r'$, and hence it does not decide $1\!-\!\veee$ up to time $1$ in $r$.
\item
If there exists a process $k \ne i$ with initial value $1\!-\!\veee$, s.t.\ $\node{k,0}$ is seen by $\node{i,1}$, then $k \notin \{i,j\}$. Hence,
as $t\!>\!0$, there exists a run $r'$ of $Q$, s.t.\ \ \emph{1)} $r'_i(1)\!=\!r_i(1)$,\ \ \emph{2)}~neither $i$ nor $j$ fail in $r'$,\ \ \emph{3)} $\node{k,0}$ is hidden from $\node{j,1}$ in $r'$, and\ \ \emph{4)} $Z_j\!=\!Z_i$ in $r'$. (Once again, $Z_i$ has the same value in both $r$ and $r'$.)
By Case I (switching the roles of $i$ and $j$), $j$ decides $\veee$ before or at time $1$ in $r'$, and therefore by \Agreement, $i$ cannot decide $1\!-\!\veee$
in $r'$, and hence it does not decide $1\!-\!\veee$ up to time $1$ in $r$.
\item
Otherwise, $\node{k,0}$ is seen by $\node{i,1}$ for all processes $k$, and $k$ has initial value $\veee$ for all processes $k \ne i$. As $|Z_i|<n$, we have that $i$ has
initial value $1\!-\!\veee$. Thus, there exists a run $r'$ of $Q$, s.t.\ \ \emph{1)} $r'_i(1)\!=\!r_i(1)$,\ \ \emph{2)}~$f=0$ in $r'$, and\ \ \emph{3)} $Z_j\!=\!Z_i$ in $r'$. (Once again, $Z_i$ has the same value in both $r$ and $r'$.)
As $i$ has initial value $1\!-\!\veee$ in $r'$ as well, by Case II (switching the roles of $i$ and $j$), $j$ decides $\veee$ before or at time $1$ in $r'$, and therefore by \Agreement, $i$ cannot decide $1\!-\!\veee$
in $r'$, and hence it does not decide $1\!-\!\veee$ up to time $1$ in $r$, and the proof is complete. \qedhere
\end{enumerate}
\end{proof}

\begin{lemma}[No Earlier Decisions]\label{maj-not-decide}
Assume that $n\!>\!2$ and $t\!>\!0$. Let $Q\dom\OptMaj$ solve Consensus and let $r$ be a run of $Q$.
Let $i$ be a process and let $m$ be a time, s.t.\ $\lnot K_i(\Maj\!=\!0)$ and $\lnot K_i(\Maj\!=\!1)$.
If there exists a hidden path w.r.t.\ $\node{i,m}$, then $i$ does not decide at $(r,m)$.
\end{lemma}

\begin{proof}
Let $v\in\{0,1\}$ be a value. We show that $i$ does not decide $\veee$ at $(r,m)$.

We first consider the case in which $m\!=\!0$.
In this case, there exists a run $r'$ of $Q$ s.t.\ \ \emph{1)} $r'_i(0)=r_i(0)$,\ \ \emph{2)} $\Maj\!=\!1\!-\!\veee$, and\ \ \emph{3)} $f=0$.
As $f=0$ and $\Maj\!=\!1\!-\!\veee$ in $r'$, we have $K_i(\Maj\!=\!1\!-\!\veee)$ at $(r',1)$, and therefore, by \cref{maj-decide-first-round}, $i$ decides $1\!-\!\veee$ before or at $1$ in $r'$; therefore, $i$ does not decide $\veee$ at $(r',0)$, and hence neither does it decide $\veee$ at $(r,0)=(r,m)$.

We turn to the case in which $m\!>\!0$.
As there exists a hidden path w.r.t.\ $\node{i,m}$, for every $0\le\ell\le m$ there exists a process $b_{\ell}$ s.t.\ $\node{b_{\ell},\ell}$ is hidden from $\node{i,m}$.
Thus, there exists a run $r'$ of $Q$ s.t.\ \ \emph{1)} $r'_i(m)=r_i(m)$,\ \ \emph{2)} $\Maj\!=\!1\!-\!\veee$,\ \ \emph{3)} $\node{b_1,1}$ sees $\node{k,0}$ for all processes $k$ (and
therefore $K_{b_1}(\Maj\!=\!1\!-\!\veee)$ at $1$,\ \ \emph{4)} $\node{b_{\ell},\ell}$ is seen by $\node{b_{\ell+1},\ell+1}$ for every $1\le\ell<m$, and\ \ \emph{5)} neither $b_m$
nor $i$ fail in $r'$.
We show by induction that $b_{\ell}$ decides $1\!-\!\veee$ before or at $\ell$ in $r'$, for every $1 \le \ell\le m$.

Base: By \cref{maj-decide-first-round}, $b_1$ decides $1\!-\!\veee$ before or at $1$ in $r'$.

Step: Let $1<\ell\le m$ and assume that $b_{\ell-1}$ decides $1\!-\!\veee$ before or at $\ell\!-\!1$ in $r'$. As $\node{b_{\ell-1},\ell\!-\!1}$ is seen by $\node{b_{\ell},\ell}$ in $r'$, there
exists a run $r''\!=\!Q[\gamma]$ of $Q$, s.t.\ \ \emph{1)} $r''_{b_{\ell}}(\ell)\!=\!r'_{b_{\ell}}(\ell)$, and\ \ \emph{2)} Neither $b_{\ell-1}$ nor $b_{\ell}$ fail in $r''$.
As $\node{b_{\ell-1},\ell\!-\!1}$ is seen by $\node{b_{\ell},\ell}$, and as $r''_{b_{\ell}}(\ell)\!=\!r'_{b_{\ell}}(\ell)$, $b_{\ell-1}$ decides $1\!-\!\veee$ before or at $\ell\!-\!1$ in $r''$ as well.
As neither $b_{\ell-1}$ nor $b_{\ell}$ fail in $r''$, by \Agreement~$b_{\ell}$ does not decide $\veee$ before or at $\ell$ in $r''$. As $\node{b_1,1}$ is seen by $\node{b_{\ell},\ell}$ in $r'$,
we have $K_{b_{\ell}}(\Maj\!=\!1\!-\!\veee)$ at $(r',\ell)$, and therefore also at $(r'',\ell)$. Thus, $b_{\ell}$ decides in $(\OptMaj[\gamma],\ell)$, and therefore $b_{\ell}$ decides
before or at $\ell$ in $r''$, and so it decides $1\!-\!\veee$ before or at $\ell$ in $r''$, and hence it also decides $1\!-\!\veee$ before or at $\ell$ in $r'$, and the proof by induction is complete.

As we have shown, $b_m$ decides $1\!-\!\veee$ in $r'$. As neither $b_m$ nor $i$ fail in $r'$, by \Agreement~$i$ does not decide $\veee$ at $(r',m)$, and therefore neither
does it decide $\veee$ at $(r,m)$.
\end{proof}
We can now prove \cref{thm:optmaj}.

\begin{proof}[Proof of \cref{thm:optmaj}]

\Agreement, \Decision\ and \Validity\ are straightforward and left to the reader. If $n\!>\!2$, then unbeatability follows from \cref{maj-not-decide}.
If $n\!=\!1$, then it is straightforward to verify that the single process always decides at time $0$, and so $\OptMaj$ cannot be improved upon.
Finally, if $n\!=\!2$, then it is easy to check that $\OptMaj$ is equivalent to $\OptZ$, and so is unbeatable.

The fact that all decisions are performed by time $f+1$ follows, exactly as in \cref{optz-f1}, from the fact that a hidden path exists w.r.t.\ each undecided process.
\end{proof}

We note that the condition $\tee\!>\!0$ in \cref{thm:optmaj} cannot be dropped if $n\!>\!2$. Indeed, if $t\!=\!0$ and $n\!>\!2$, then both $\OptZ$ and $\OptO$ (in which some decisions are made at time $0$, and the rest --- at time $1$)
strictly dominate $\OptMaj$ (in which all decisions are made at time $1$).

\subsection{Uniform Consensus}
\label{sec-uni-cons-proofs}

We note that while the assumption $\tee\!<\!n$ simplifies presentation throughout the proofs below, the case $t\!=\!n$ can be analysed via similar tools.

\begin{proof}[Proof of Lemma~\ref{lem:correct-uni}]

Let~$P$ be a uniform consensus protocol, and let $r$ be a run of~$P$ such that
$(R_P,r,m) \not\sat K_i\cv$.
Thus, there exists a run $r'\in P[\alpha']$ such that $r_i(m)=r'_i(m)$ and
$(R_P,r',m) \not\sat \cv$.
Consider the adversary $\beta$ that agrees with~$\alpha'$ up to time~$m$, and
in which all active but faulty processes at $(r',m)$ crash at time~$m$ without
sending any messages. 
$\beta\in\gammacr$ because it has a legal input vector (identical to~$\alpha'$),
and at most~$\tee$ crash failures, as it has the same set of faulty processes
as~$\alpha'\in\gammacr$. It follows that $r''=P[\beta]$ is a run of~$P$.
Since $\beta$ agrees with $\alpha'$ on the first~$m$ rounds, we have that
$r''_i(m)=r'_i(m)$. Nonetheless, no correct process will ever know $\existsv$
in~$r''$, and thus by \Validity\ no correct process ever decides $\veee$ in~$r''$.
By decision, all correct processes thus decide not on $\veee$. By \UniAg, and as
$t\!<\!n$ (i.e.\ there are correct processes),
$i$ cannot decide on $\veee$ in~$r''$, and thus, as $r''_i(m)=r'_i(m)=r_i(m)$,
it cannot decide on $\veee$ in $r$ at $m$.
\end{proof}

Before moving on to prove \cref{lem:u-know}. We first introduce some notation.

\begin{definition}
For a node $\node{i,m}$, we denote by $\knownf{i,m} \in \{0,\ldots,t\}$ the number of failures known to $\node{i,m}$, i.e.\
the number of processes $j \ne i$ from which $i$ does not receive a message at time $m$.
\end{definition}

We note that \defemph{d}, as defined in \cref{lem:u-know}, is precisely $\knownf{i,m}$.

\begin{proof}[Proof of Lemma~\ref{lem:u-know} (Sketch)]

It is straightforward to see that each of conditions (a) and (b) implies $K_i\cv$ (Condition (a): as $\node{i,m-1}$ is seen at $m$ by all correct processes; condition (b): as the number of distinct
processes knowing $\exists0$, \emph{including $i$ itself}, is greater than the maximum number of active processes that can yet fail).
If neither condition holds, then $i$ considers it possible that only incorrect processes know $\existsv$, and that they all  immediately fail
($i$ at time $m$ before sending any messages, and the others ---
immediately after sending the last message seen by $i$),
in which case no correct process would ever know $\existsv$.
\end{proof}

As with~$\Pz$ in the case of consensus, by analysing decisions in protocols dominating $\UPz$, we show that no Uniform Consensus protocol can dominate
$\UOptZ$. \cref{decide-when-0-first-round,decide-when-0} give sufficient conditions for deciding~$0$ in any Uniform Consensus protocol
dominating $\UPz$. As mentioned above, the analysis is considerably subtler 
for Uniform Consensus, because the analogue of \cref{lem:decide-when-0} is not true. Receiving a message with value~0 in a protocol dominating~$\UPz$ does not imply that the sender has decided~0. 

\begin{lemma}[No decision at time $0$]\label{undecided-at-0}
Assume that  $t\!>\!0$.
Let $Q$ solve Uniform Consensus. No process decides at time $0$ in any run of $Q$.
\end{lemma}

\begin{proof}
As $t\!<\!n$, by \cref{lem:correct-uni} it is enough to show that $\lnot K_i \existsv$ for every process $i$ and $v \in \{0,1\}$.
As $0\!<\!t$, and as $\knownf{i,0}=0$ for all processes $i$ by definition, we have that by \cref{lem:u-know}, the proof is complete.
\end{proof}

\begin{lemma}[Decision at time $1$]\label{decide-when-0-first-round}
Let $Q\dom\UPz$ solve Uniform Consensus and let $r=r[\alpha]$ be a run of $Q$. Let $i$ be a process with initial value~$0$ in $r$ s.t. $i$ is  active 
at time $1$ in $r$.
If either of the following hold in $r$, then $\node{i,1}$ decides~$0$ in~$r$.
\begin{parts}
\item\label{decide-when-0-first-round-more-zeros}
$t>0$ ~and~ there exists a process $j \ne i$ with initial value $0$ s.t.\ $\node{j,0}$ is seen by $\node{i,1}$.
\item\label{decide-when-0-first-round-no-more-zeros}
$t>1$ ~and~ $\knownf{i,1} < t$.
\end{parts}
\end{lemma}

\begin{proof}
For both parts, we first note that by \cref{lem:u-know} and by definition of $\UPz$, $i$ decides $0$ at $(\UPz[\alpha],1)$. As $Q\dom\UPz$, we thus have that $i$ must decide upon some
value in $r$ by time $1$.
By \cref{undecided-at-0}, $i$ does not decide at $(r,0)$. Thus, $i$ must decide at $(r,1)$.

We now show \cref{decide-when-0-first-round-more-zeros} by induction on $n\!-\!|Z^0_i|$,
where $Z^0_i$ is defined to be the set of processes $k$ with initial value $0$, s.t.\ $\node{k,0}$ is seen by $\node{i,1}$. Note that by definition, $i,j \in Z^0_i$, and
so $1 < |Z^0_i| \le n$.

Base: $|Z^0_i|=n$. In this case, all initial values are $0$, and so by \Validity~$i$ decides $0$ at $(r,1)$.

Step: Let $1<\ell<n$ and assume that \cref{decide-when-0-first-round-more-zeros} holds whenever $|Z^0_i|=\ell+1$. Assume that $|Z^0_i|=\ell$. We reason by cases.

\begin{enumerate}[label=\Roman*.]
\item
If there exists a process $k$ s.t.\ $\node{k,0}$ is hidden from $\node{i,1}$, then there exists a run $r'$ of $Q$, s.t.~~\emph{1)} $r'_i(1)\!=\!r_i(1)$,\ \ \emph{2)}~$j$ is active 
at $(r',1)$,~~\emph{3)} $k$ has initial value $0$ in $r'$, and~~~\emph{4)} $Z^0_j = Z^0_i\!\cup\!\{k\}$ in $r'$. (Note that by definition, $Z^0_i$ has the same value in both $r$ and $r'$.)
By the induction hypothesis (switching the roles of $i$ and $j$), $j$ decides $0$ at $(r',1)$, and therefore by \UniAg, $i$ cannot decide $1$ at $(r',1)$, and
hence it does not decide $1$ at $(r,1)$. Thus, $i$ decides~$0$ at~$(r,1)$.
\item
Otherwise, $\node{k,0}$ is seen by $\node{i,1}$ for all processes $k$. As $|Z^0_i|<n$, there exists a process $k \notin Z^0_i$ (in particular, $k \notin \{i,j\}$). Hence,
as $t\!>\!0$, there exists a run $r'$ of $Q$, s.t.\ \ \emph{1)} $r'_i(1)\!=\!r_i(1)$,\ \ \emph{2)}~$j$ is active
at $(r',1)$,\ \ \emph{3)} $\node{k,0}$ is hidden from $\node{j,1}$ in $r'$, and~~\emph{4)} $Z^0_j=Z^0_i$ in $r'$. (Once again, $Z^0_i$ has the same value in both $r$ and $r'$.)
By Case I (switching the roles of $i$ and $j$), $j$ decides $0$ at~$(r',1)$, and therefore by \UniAg, $i$ cannot decide $1$
at $(r',1)$, and hence it does not decide $1$ at $(r,1)$. Thus, $i$ decides~$0$ at~$(r,1)$.
\end{enumerate}

We move on to prove \cref{decide-when-0-first-round-no-more-zeros}.
If $\node{k,0}$ is hidden from $\node{i,1}$ for all processes $k \ne i$, then $\lnot K_i\exists1$ at $(r,1)$.
Thus, by \cref{lem:correct-uni}, $i$ cannot decide $1$ at $(r,1)$, and so must decide $0$ at $(r,1)$.
Otherwise, there exists a process $k\ne i$ s.t.\ $\node{k,0}$ is seen by $\node{i,1}$. As $n\!>\!t\!>\!1$, we have $n\!>\!2$ and so
there exists a process
$j \notin \{i,k\}$; if $\knownf{i,1}>0$, then we pick $j$ s.t.\ $\node{j,0}$ is hidden from $\node{i,1}$. Since
$t>1$ (for the case in which  $\knownf{i,1}=0$ and $\node{j,0}$ is seen by $\node{i,1}$) and since $t>\knownf{i,1}$ (for the case in which $\node{j,0}$ is
hidden from $\node{i,1}$), there exists a run $r'$ of $Q$,s.t.\ \ \emph{1)}~$r'_i(1)\!=\!r_i(1)$,\ \ \emph{2)}~$k$ never fails in $r'$,\ \ \emph{3)}~$j$ fails at $(r',0)$ before sending any messages except perhaps to $i$, and\ \ \ \emph{4)}~$i$ fails at $(r',1)$, immediately after deciding but before sending any messages.
Thus, there exists a run $r''$ of $Q$, s.t.\ \ \emph{1)}~$r''_k(m')\!=\!r'_k(m')$ 
\underline{for all} $m'$, 
\ \ \emph{2)} $k$ never fails in $r''$,\ \ \emph{3)}~$i$ and~$j$ both have initial value $0$ in $r''$,\ \ \emph{4)} $j$ fails at $(r'',0)$ while successfully sending a message only to $i$ (and therefore $j \in Z^0_i$ in $r''$), and\ \ \emph{5)} $i$ fails at $(r'',1)$, immediately after deciding but before sending out any messages.
By \cref{decide-when-0-first-round-more-zeros}, $i$ decides $0$ at $(r'',1)$, and therefore $k$ can never decide $1$ during $r''$, and therefore
neither during~$r'$. As $k$ never fails during $r'$, by \Decision\ it must thus decide~$0$ at some point during $r'$. Therefore, by \UniAg, $i$ cannot decide $1$ at $(r',1)$, and thus it does not decide $1$ at $(r,1)$. Thus, $i$ decides $0$ at $(r,1)$.
\end{proof}

\begin{lemma}[Decision at times later than $1$]\label{decide-when-0}
Let $Q\dom\UPz$ solve Uniform Consensus, let $r\!=\!Q[\alpha]$ be a run of $Q$ and let $m\!>\!0$.
Let $i$ be a process s.t.\ $K_i\exists0$ holds at time $m$ for the first time in~$r$, s.t.\ $K_i\cz$ holds at time $m+1$ for the first time in~$r$, and s.t. $i$ is  active 
at $(r,m+1)$.
If either of the following hold in $r$, then $i$ decides~$0$ at~$(r,m+1)$.
\begin{parts}
\item\label{decide-when-0-hidden-z}
All of the following hold.
\begin{itemize}
\item
$\knownf{i,m+1}<t$.
\item
There exists a process $z$ s.t.\ $K_z\exists0$ holds at time $m\!-\!1$, s.t.\ $\node{z,m\!-\!1}$ is seen by $\node{i,m}$, but s.t.\ $\node{z,m}$ is not seen by $\node{i,m\!+\!1}$, 
\item
There exists a process $j \ne i$
s.t.\ $\node{j,m}$ is seen by $\node{i,m\!+\!1}$ and $\node{z,m\!-\!1}$ is seen by $\node{j,m}$.
\end{itemize}
\item\label{decide-when-0-low-knownf}
$\knownf{i,m+1}<t-1$.
\end{parts}
\end{lemma}

\begin{proof}
We prove the lemma by induction on $m$, with the base and the step sharing the same proof (as will be seen below, the conceptual part of an induction base will be played, in a sense, by \cref{decide-when-0-first-round}).

We prove both parts together, highlighting local differences in reasoning for the different parts as needed. For \cref{decide-when-0-low-knownf}, we denote by $z$ an arbitrary process s.t.\ $K_z\exists0$ holds at time $m-1$ and s.t. $\node{z,m\!-\!1}$ is seen by $\node{i,m}$. (As $m>0$, such a process must exist for $i$ to know $\exists0$ at time $m$ for the first time; nonetheless, unlike when proving \cref{decide-when-0-hidden-z},  it is not guaranteed when proving this part that $\node{z,m}$ is not seen by $\node{i,m\!+\!1}$.)

We first note that by \cref{lem:correct-uni} and by definition of $\UPz$, $i$ decides $0$ at $(\UPz[\alpha],m\!+\!1)$. As $Q\dom\UPz$, we thus have that $i$ must decide upon some value in $r$ by time $m\!+\!1$. By \cref{lem:correct-uni}, the precondition for deciding $0$ is not met by $i$ at $(r,m)$. Therefore, it is enough to show
that $i$ does not decide $1$ before or at time $m\!+\!1$ in $r$ in order to show that $i$ decides $0$ at $(r,m\!+\!1)$.

Let $Z^{z,m}_i$ be the set of processes $k$ s.t.\ $\node{k,m}$ is seen by $\node{i,m\!+\!1}$ in $r$ and s.t.\ $\node{z,m-1}$ is seen by $\node{k,m}$ in $r$. (By definition,
$i \in Z^{z,m}_i$.)
Let $C_i$ be the set of all processes $k$ s.t.\ $\node{k,m}$ is either seen by, or hidden from $\node{i,m\!+\!1}$ (i.e.\ the set of nodes that $\node{i,m\!+\!1}$ does not know to be inactive
at time $m$). Note that by definition, $Z^{z,m}_i\subseteq C_i$.
We first consider the case in which $Z^{z,m}_i\supsetneq\{i\}$, and prove the $m$-induction step (for the given $m$) for this case by induction on $|C_i \setminus Z^{z,m}_i|$.

Base: $Z^{z,m}_i=C_i$. In this case, $\node{i,m\!+\!1}$ does not know that $z$ fails 
at time $m\!-\!1$ .
Thus, 
$z \in C_i$ and therefore $z \in Z^{z,m}_i$. 
It follows that $\node{z,m}$ is seen by
$\node{i,m\!+\!1}$ and therefore
the second condition of \cref{decide-when-0-hidden-z} does not hold. Thus, the condition of 
\cref{decide-when-0-low-knownf} holds: 
$\knownf{i,m\!+\!1}  < t\!-\!1$. Furthermore, we thus have that $z$ is active 
at time $m$. We now argue that $z$ decides $0$
at $(r,m)$, which completes the proof of the base case,
as by \UniAg~$i$ can never decide $1$ during $r$. We reason by cases; for both cases, note that since $\node{z,m}$ is seen by $\node{i,m\!+\!1}$, 
we have that 
$\knownf{z,m} \le \knownf{i,m\!+\!1}  < t\!-\!1$.

\begin{itemize}
\item
If $m=1$: As $K_z\exists0$ at time $m\!-\!1=0$, $z$ has initial value $0$.
As $\knownf{z,m}<t\!-\!1$, we have that $t>1$.
By \cref{decide-when-0-first-round-no-more-zeros} of \cref{decide-when-0-first-round} (for $i=z$), we thus have that $z$ decides $0$ at $(r,1)=(r,m)$.
\item
Otherwise, $m\!>\!1$. In this case, as $\node{z,m\!-\!2}$ is seen by $\node{i,m\!-\!1}$, and as $K_i\exists0$ holds at time $m$ for the first time, we have that $K_z\exists0$ holds at time
$m\!-\!1$ for the first time. Similarly, as $\node{z,m\!-\!1}$ is seen by $\node{i,m}$, and as $K_i\cz$ does not hold at time $m$, we have that $K_z\cz$ does not hold at time $m\!-\!1$.
By \cref{decide-when-0-low-knownf} of the $m$-induction hypothesis (for $i=z$), $z$ decides $0$ at $(r,m)$.
\end{itemize}

Step: Let $\{i\}\subsetneq Z^{z,m}_i \subsetneq C_i$, and assume that the claim holds whenever $Z^{z,m}_i$ is of larger size.
For \cref{decide-when-0-hidden-z}, note that $j \in Z^{z,m}_i$, for $j$ as defined in the conditions for that part; for \cref{decide-when-0-low-knownf}, let $j \in Z^{z,m}_i$ be arbitrary.
Analogously to the proof of the induction step in the proof of \cref{decide-when-0-first-round-more-zeros} of \cref{decide-when-0-first-round}, we reason by cases. For the time being, assume that the conditions
of \cref{decide-when-0-low-knownf} hold, i.e.\ that $\knownf{i,m\!+\!1}<t\!-\!1$.

\begin{enumerate}[label=\Roman*.]
\item
If there exists a process $k \in C_i$ s.t.\ $\node{k,m}$ is hidden from $\node{i,m\!+\!1}$, then there exists a run $r'$ of $Q$, s.t.\ \ \emph{1)} $r'_i(m\!+\!1)=r_i(m\!+\!1)$,\ \ \emph{2)}~$j$ is active 
at $(r',m\!+\!1)$,\ \ \emph{3)}~$\node{z,m-1}$ is seen by $\node{k,m}$ in $r'$, and\ \ \emph{4)}~$Z^{z,m}_j = Z^{z,m}_i\!\cup\!\{k\}$ and $C_j=C_i$ in $r'$. (Note that by definition, $Z^{z,m}_i$ and $C_i$ have the same values in both~$r$ and~$r'$.)
We note that $\knownf{j,m\!+\!1}=\knownf{i,m\!+\!1}-1$ in~$r'$, and that by definition $\knownf{i,m\!+\!1}$ is the same in both $r$ and $r'$.
By the 
inductive hypothesis for $Z^{z,m}_j$ (i.e., for $j$ w.r.t.~$z$ at time~$m$), 
$j$ decides $0$ at $(r',m\!+\!1)$, and therefore by \UniAg, $i$ cannot decide $1$ in $r'$,
and therefore it cannot decide $1$ before or at $m\!+\!1$ in $r'$, and the proof is complete.

\item
Otherwise, for each process $k \in C_i$, $\node{k,m}$ is seen by $\node{i,m\!+\!1}$. As $Z^{z,m}_i\subsetneq C_i$,
there exists a process~$k \ne i$ s.t.\ $\node{k,m}$ is seen by $\node{i,m\!+\!1}$ but s.t.\ $\node{z,m\!-\!1}$ is hidden from $\node{k,m}$ (thus $k \ne j$). Hence, and since $\knownf{i,m\!+\!1}<t$, there exists a run $r'$ of $Q$, s.t.\ \ \emph{1)} $r'_i(m\!+\!1)=r_i(m\!+\!1)$,\ \ \emph{2)}~$j$ is active 
at $(r',m\!+\!1)$,\ \ \emph{3)}~$\node{k,m}$ is hidden from $\node{j,m\!+\!1}$ in~$r'$, and\ \ \emph{4)}~$Z^{z,m}_j=Z^{z,m}_i$ and $C_j\supseteq C_i$ in $r'$. (Once again, $Z^{z,m}_i$ and $C_i$ have the same values
in both $r$ and $r'$.) We note that $\knownf{j,m+1}=\knownf{i,m+1}+1$ in $r'$, and that once more, by definition, $\knownf{i,m+1}$ is the same in both $r$ and $r'$.
By Case I (for~$i=j$), and since Case I uses the 
inductive hypothesis for $Z^{z,m}_j$ with one less failure, we conclude that 
$j$ decides~$0$ at $(r',m\!+\!1)$.
Therefor, 
by \UniAg, $i$ cannot decide $1$ at $(r',m\!+\!1)$, and thus it cannot decide~$1$ before or at $m+1$ in $r$, and the proof is complete.
\end{enumerate}

To show that the $Z^{z,m}_i$-induction step also holds under the conditions of \cref{decide-when-0-hidden-z}, we observe that since $\node{z,m}$  is not seen by $\node{i,m\!+\!1}$ in this case, the amount of invocations of Case II
(which uses Case I with one additional known failure) before reaching the $Z^{z,m}_i$-induction base is strictly smaller than that of Case I (which uses the  $Z^{z,m}_i$-induction hypothesis with
one less known failure), and therefore the $Z^{z,m}_i$-induction base is reached with less known failures, i.e.\ with less than $t-1$ known failures, i.e.\ the conditions of \cref{decide-when-0-low-knownf} hold at that point.

Finally, we consider the case in which $Z^{z,m}_i=\{i\}$. As any $j$ as in \cref{decide-when-0-hidden-z} satisfies $j \in Z^{z,m}_i$, we have that the conditions
of \cref{decide-when-0-low-knownf} hold, i.e.\ $\knownf{i,m\!+\!1}<t\!-\!1$. Furthermore, in we have that $\node{z,m}$ is not seen by $\node{i,m\!+\!1}$ (otherwise, $z \in Z^{z,m}_i$). As $\knownf{i,m\!+\!1}<t\!-\!1<n\!-\!2$, there exist two distinct processes $j,k \ne i$ that are not known to $\node{i,m\!+\!1}$ to fail (and thus $i,j,k,z$ are distinct). Thus, $\node{j,m}$ and
$\node{k,m}$ are seen by $\node{i,m\!+\!1}$.

By definition of $j,k$, there exists a run $r'$ of $Q$, s.t.\ \ \emph{1)} $r'_i(m\!+\!1)=r_i(m\!+\!1)$,\ \ \emph{2)}~$k$ never fails in $r'$,\ \ \emph{3)}~$j$ fails at $(r',m)$ before sending any messages,\ \ \emph{4)} $i$ fails at $(r',m+1)$, immediately after deciding but before sending any messages, and\ \ \emph{5)} the faulty processes in~$r'$ are those known by $\node{i,m}$ to fail in $r$, and in addition $i$ and $j$. We note that by definition, $\knownf{i,m\!+\!1}$ is the same in $r$ and $r'$, even though the
number of failures in $r'$ is $\knownf{i,m\!+\!1}+2$.
We notice that there exists a run $r''$ of $Q$, s.t.\ \ \emph{1)}~$r''_k(m')=r''_k(m')$ \underline{for all} $m'$,\ \ \emph{2)} $k$ never fails in $r''$,\ \ \emph{3)} $\node{z,m-1}$ is seen by both $\node{i,m}$ and $\node{j,m}$ in $r''$,\ \ \emph{4)} $j$ fails at $(r'',m)$ while successfully sending a message only to $i$ (and therefore both $j \in Z^{z,m}_i$ and $\knownf{i,m+1}<t-1$ in $r''$), and\ \ \emph{5)} $i$ fails at $(r'',m+1)$, immediately after deciding but before sending out any messages.
By the proof for the case in which $Z^{z,m}_i\supsetneq\{i\}$ ($j\in Z^{z,m}_i$), $i$ decides~$0$ at $(r'',m\!+\!1)$, and therefore $k$ can never decide $0$ during $r''$, and therefore neither during~$r'$. As $k$ never fails during~$r'$, by \Decision\ it must thus decide $0$ at some point during~$r'$. Therefore, by \UniAg, $i$ cannot decide~$1$ before or at $m\!+\!1$ in $r'$, and thus it does not decide~$1$ before or at $m+1$ in $r$, and the proof is complete.
\end{proof}

Now that we have established when processes must decide $0$ in any protocol dominating $\Pz$, we can deduce when processes cannot decide in any such protocol.

\begin{lemma}[No Earlier Decisions when $K_i\exists0$]\label{no-earlier-k0}
Let $Q\dom\UPz$ solve Uniform Consensus, let $r$ be a run of $Q$, let $m$ be a time, and let $i$ be a process.
If at time $m$ in $r$ we have $K_i\exists0$, but $\lnot K_i\cz$, then $i$ does not decide at $(r,m)$.
\end{lemma}

\begin{proof}
If $m\!=\!0$, then by \cref{lem:u-know} and since $\lnot K_i\cz$ at $m\!=\!0$ (even though $K_i\exists0$),
we have $t\!>\!0$. Thus, by \cref{undecided-at-0}, $i$ does not decide at $(r,m)$. Assume henceforth, therefore, that $m\!>\!0$.

As $\lnot K_i\cz$, we have that by \cref{lem:u-know}, $\lnot K_i\exists0$ at time $m\!-\!1$. Thus, there exists a process $z$ s.t.\
$K_z\exists0$ at~$m\!-\!1$, and $\node{z,m\!-\!1}$ is seen by $\node{i,m}$. In turn, by \cref{lem:u-know}, we have that $\knownf{i,m}<t-1$.
There exists a run~$r'$ of $Q$, s.t.\ \ \emph{1)} $r'_i(m)\!=\!r_i(m)$, and\ \ \emph{2)} the faulty processes in $r'$ are those known by~$\node{i,m}$ to fail in $r$. We henceforth reason about $r'$. By definition of $r'$, $\knownf{i,m\!+\!1}=\knownf{i,m}<t\!-\!1$
(by definition, the value of $\knownf{i,m}$ is the same in both $r$ and $r'$). Thus, by \cref{decide-when-0-low-knownf} of \cref{decide-when-0}, $i$ decides $0$ at $(r',m\!+\!1)$, and hence $i$ does not decide at $(r',m)$,
and therefore neither does it decide at $(r,m)$.
\end{proof}

\begin{lemma}[No Earlier Decisions when $\lnot K_i\exists0$]\label{no-earlier-k1}
Assume that  $t\!>\!0$.
Let $Q\dom\UPz$ solve Uniform Consensus, let $r$ be a run of $Q$, let $m$ be a time, and let $i$ be a process.
If there exists a hidden path w.r.t. $\node{i,m}$ in $r$, and if at time $m$ in $r$ we have $\lnot K_i\exists0$, then $i$ does not
decide at $(r,m)$.
\end{lemma}

\begin{proof}
As $\lnot K_i\exists0$ at time $m$, then by \Validity, $i$ does not decide $0$ at $(r,m)$. Thus, it is enough to show that $i$ does not decide $1$ at $(r,m)$ in order to complete the proof.
If $m\!=\!0$, then by \cref{undecided-at-0}, $i$ does not decide $1$ at $(r,m)$ either. Assume henceforth, therefore, that $m\!>\!0$.

As there exists a hidden path w.r.t. $\node{i,m}$, there exist processes $z,j \ne i$ s.t.\ $\node{z,m\!-\!1}$ is hidden from $\node{i,m}$ and s.t.\
$\node{j,m\!-\!1}$ is seen by $\node{i,m}$.

We first consider the case in which $\knownf{i,m}<t$.
In this case, there exists a run $r'\!=\!Q[\beta]$ of $Q$, s.t.\ all of the following hold in $r'$:
\begin{itemize}
\item
$r'_i(m)=r_i(m)$.
\item
$z$ is the unique process that knows $\exists0$ at $m\!-\!1$, and knows so then for the first time, either having initial value $0$ (if $m\!=\!1$) or (as explained in the Nonuniform Consensus section) seeing only a single node that knows $\exists0$ at $m\!-\!2$ (if $m\!>\!1)$.
\item
$z$ fails at $(r',m\!-\!1)$, successfully sending messages to all nodes except for $i$.
\item
The faulty processes in $r'$ are those known by $\node{i,m}$ to fail in $r$, and in addition $i$, which fails at time $m$ without
sending out any messages. In particular, $j$ never fails.
\end{itemize}

We henceforth reason about $r'$. First, we note that $\node{j,m\!+\!1}$ does not know that $z$ fails at $m\!-\!1$ (as opposed to at~$m$). As $\node{j,m}$ sees $\node{z,m\!-\!1}$,
as $K_z\exists0$ at $m\!-\!1$, and as $j$ never fails, by \cref{lem:u-know}
we have that $K_j\cz$ at $(r',m\!+\!1)$. Thus, $j$ decides at $(\UPz[\beta],m\!+\!1)$, and so $j$ must decide
before or at~$m\!+\!1$ in $r'$. As $r_i(m)\!=\!r'_i(m)$, then by \UniAg\ it is enough to show that~$j$ does not decide $1$ up to time
$m+1$ in $r'$ in order to complete the proof.

There exists a run $r''$ of $Q$, s.t.\ \ \emph{1)} $r''_j(m\!+\!1)=r'_j(m\!+\!1)$, and\ \ \emph{2)}~the only difference between $r''$ and $r'$ up to time $m$ is that in $r''$, $z$ fails only at time $m$, after 
deciding 
but without sending a message to $j$. By \UniAg, it is enough to show that $z$ decides $0$ at $(r'',m)$ in order to complete the proof.

We henceforth reason about $r''$. As $z$ does not know at $m$ that neither $z$ nor $i$ fail, we have $\knownf{z,m\!-\!1}\le\knownf{z,m}<t\!-\!1$.
Thus, $t\!>\!1$. If $m\!=\!1$, we therefore have by \cref{decide-when-0-first-round-no-more-zeros} of \cref{decide-when-0-first-round} that $z$ decides $0$ at $(r'',m)$. Otherwise, $m\!>\!1$.
As $K_z\exists0$ at $m\!-\!1$ for the first time, as $\node{z,m\!-\!1}$ sees only one node at $m\!-\!1$ that knows $\exists0$, and as $\knownf{z,m}<t\!-\!1$, by \cref{lem:u-know} we have $\lnot K_z\cz$ at $m\!-\!1$. Thus, by \cref{decide-when-0-low-knownf} of \cref{decide-when-0} (for $i=z$), $z$ decides $0$ at $(r'',m)$. Either way, the proof is complete.

We now consider the case in which $\knownf{i,m}=t$.
There exists a run $r'\!=\!Q[\beta]$ of $Q$, s.t.\ all of the following hold:
\begin{itemize}
\item
$r'_i(m)=r_i(m)$.
\item
All processes $k$ s.t.\ $\node{k,m\!-\!1}$ is hidden from $\node{i,m}$ (including $k=z$) know $\exists0$ at $(r',m\!-\!1)$, either having initial value $0$ (if $m\!=\!1$) or all seeing only a single node that knows $\exists0$ at $m\!-\!2$ (and which fails at time $m\!-\!2$ without being seen by $\node{i,m}$) --- denote this node by $z'$.
\item
All such processes fail at time $m\!-\!1$, successfully sending messages to all nodes except for $i$.
\item
The faulty processes failing in $r'$ are those known by $\node{i,m}$ to fail in $r$. In particular, there are $t$ such processes.
\end{itemize}
We henceforth reason about $r'$.
We note that as $i$ never fails, $\knownf{i,m\!-\!1}\le\knownf{j,m}$ (equality can actually be shown to hold here, but we do not need it).
As the number of nodes at $m\!-\!1$ knowing $\exists0$ that are seen by $\node{j,m}$ equals $\knownf{i,m}-\knownf{i,m\!-\!1}\ge t-\knownf{j,m}$ (by the above
remark, equality holds here as well), we have by \cref{lem:u-know} that $K_j\cz$ at $m$, and therefore $j$ decides
at $(\UPz[\beta],m)$; thus, it must decide before or at $m$ in $r'$. As $r_i(m)\!=\!r'_i(m)$, by \UniAg\ it is enough
to show that $j$ does not decide $1$ up to time $m$ in $r'$ in order to complete the proof.

We proceed with an argument similar in a sense to those of \cref{decide-when-0-first-round-more-zeros} of \cref{decide-when-0-first-round} and the inner induction in the proof of \cref{decide-when-0}.

As $\node{z,m\!-\!1}$ is seen by $\node{j,m}$, there exists a run $r''$ of $Q$, s.t.\ \ \emph{1)} $r''_j(m)\!=\!r'_j(m)$, and\ \ \emph{2)} the only difference between $r''$ and $r'$ up to time $m$
is that in $r'$, $z$ never fails, but rather $i$ fails at $m\!-\!1$ after sending a message to $j$ but without sending a message to $z$.
We note that there are $t$ processes failing throughout $r''$. 
We henceforth reason about $r''$. If $m\!=\!1$, then $z$ has initial value $0$ and if $m\!>\!1$, then $\node{z,m\!-\!1}$ sees $\node{z',m\!-\!2}$;
either way, by \cref{lem:u-know}, $K_z\cz$ at $(r'',m)$ and therefore $z$ must decide before or at time $m$. Thus, it is enough to show that $z$ does not decide $1$ up to time $m$ in $r''$
in order to complete the proof.

As $\node{i,m\!-\!1}$ is not seen by $\node{z,m}$, there exists a run $r'''$ of $Q$, s.t.\ \ \emph{1)} $r'''_z(m)\!=\!r''_z(m)$, and\ \ \emph{2)} the only difference between $r'''$ and $r''$ up to time $m$
is that in $r'''$, $\node{i,m-1}$ sees $\node{z',m\!-\!2}$ (or, if $m=1$, then the difference is that $i$ has initial value $0$); we note that $\node{i,m\!-\!1}$ is still seen by $\node{j,m}$.
We note that there are $t$ processes failing throughout $r'''$.
Observe that the number of nodes at $m\!-\!1$ knowing $\exists0$ that are seen by $\node{j,m}$ in $r'''$  is greater than in $r'/r''$ (between which $j$ at $m$ cannot distinguish), however $\knownf{j,m}$ remains the same between $r'/r''$ and~$r'''$; thus, $K_j\cz$ at $m$ in $r'''$ as well, and therefore $j$ must decide before or at time $m$ in $r'''$. Thus, it is enough to show that $j$
does not decide $1$ up to time $m$ in $r'''$ in order to complete the proof. We henceforth reason about~$r'''$. 

As $\node{i,m\!-\!1}$ is seen by $\node{j,m}$, there exists a run $r''''$ of $Q$, s.t.\ \ \emph{1)} $r''''_j(m)=r'''_j(m)$, and\ \ \emph{2)} the only difference between $r''''$ and $r'''$ up to time $m$
is that in $r''''$, $i$ does not fail (and is thus seen by $\node{z,m}$).
We note that there are~$t-1$ processes failing throughout $r''''$, and thus in particular $\knownf{z,m}<t$. If $m=1$, then by \cref{decide-when-0-first-round-more-zeros} of \cref{decide-when-0-first-round} (for~$i=z$ and $j=i$), $z$ decides $0$ in $(r'''',m)$. Otherwise, i.e.\ if $m\!>\!1$, by \cref{decide-when-0-hidden-z} of \cref{decide-when-0} (for $i=z$, $z=z'$, and~$j=i$), $z$ decides $0$ in $(r'''',m)$. Either way, the proof is complete.
\end{proof}

From \cref{no-earlier-k0,no-earlier-k1}, we deduce sufficient conditions for unbeatability of Uniform Consensus protocols dominating $\UPz$; these conditions also become necessary
if it can be shown that there exists some Uniform Consensus protocol dominating $\UPz$ that meets them, as we indeed show momentarily for $\UOptZ$.

\begin{lemma}\label{cor:uni-opt}
Assume that  $0<t<n$.
A protocol $Q\dom\UPz$ that solves Uniform Consensus and in which a node $\node{i,m}$ decides whenever any of the following hold at $m$, is an unbeatable Uniform Consensus protocol.
\begin{itemize}
\item
$K_i\cz$.
\item
No hidden path w.r.t.\ $\node{i,m}$ exists, and $\lnot K_i\exists0$.
\end{itemize}
\end{lemma}

\begin{proof}
Directly from \cref{no-earlier-k0,no-earlier-k1}.
\end{proof}

By \cref{cor:uni-opt}, we have that
if $\UOptZ$ solves Uniform Consensus, then it does so in an unbeatable fashion.

\begin{lemma}
$\UOptZ\dom\UPz$
\end{lemma}

\begin{proof}
As explained above, at time $\tee+1$ no hidden paths exist (see the proofs of \cref{optz-f1,thm:optmaj}), and furthermore, by \cref{lem:knowing0} we have at time $\tee+1$ that $K_i\exists0$ iff $K_i\cz$. The claim therefore holds by definition of $\UOptZ$ and $\UPz$.
\end{proof}

\begin{theorem}
\label{u-solve}
$\UOptZ$ ~solves Uniform Consensus in $\gammacr$. Furthermore,
\begin{itemize}
\item
If $f \ge t-1$, then all decisions are made by time $f+1$ at the latest.
\item
Otherwise, all decisions are made by time $f+2$ at the latest.
\end{itemize}
 \end{theorem}

\begin{proof}

\Decision:
In some run of $\UOptZ$, let $i$ be a process and let $m$ be a time s.t.\ $i$ is active at $m$ but has not decided until $m$, inclusive.
Let $\tilde{m}\le m$ be the latest time not later than $m$ s.t.\ a hidden path exists w.r.t.\ $\node{i,\tilde{m}}$. We claim that as $i$ is undecided at $m$, we have $\tilde{m} \ge m-1$;
indeed, otherwise, by $i$ being undecided at $\tilde{m}+1$ despite the absence of a hidden path w.r.t.\ $\node{i,\tilde{m}+1}$, we would have $K_i\exists0$ at $\tilde{m}+1$, and so,
by \cref{lem:u-know}, we would have $K_i\cz$ at $\tilde{m}+2\le m$ --- a contradiction to $i$ being undecided at $m$.

As a hidden path exists w.r.t.\ $\node{i,\tilde{m}}$, we have, as in the proofs of \cref{optz-f1,thm:optmaj}, that $\tilde{m}\le f$; in fact, the same proof shows the even stronger claim $\tilde{m}\le\knownf{i,\tilde{m}}$ --- we we will later return to this inequality. As $\tilde{m}\le f$, we therefore have that $m\le\tilde{m}+1\le f+1$.
We thus have that every process that is active at time $f+2$, decides by this time at the latest.

Before moving on to show \Validity\ and \UniAg, we first complete the analysis of stopping times. 
Assume that $m=f+1$.
($i$ is still a process that is active but undecided at $m$.)
As $f = m-1 \le \tilde{m} \le \knownf{i,\tilde{m}} \le  \knownf{i,m}\le f$, we
we have that both $\tilde{m}=m-1$ and $\knownf{i,m}=f$.
As $\tilde{m}=m-1$, we have that no hidden path exists w.r.t.\ $\node{i,m}$. As $i$ is undecided at $m$,
we thus have, by definition of $\UOptZ$, that $K_i\exists0$ while $\lnot K_i\cv$ at $m$.
We therefore have that $K_i\exists0$ at $m$ for the first time.
Therefore, as $m>\tilde{m}\ge0$, there exists a process $j$ such that $K_j\exists0$ at $m-1$ and s.t.\ $\node{j,m-1}$ is seen by $\node{i,m}$. Thus, by \cref{lem:u-know}
and since $\lnot K_i\cv$, we have $\knownf{i,m}<t-1$, and so $f=\knownf{i,m}<t-1$.

We thus have that if $f=t-1$, then every process that is active at time $f+1$ decides by this time at the latest.

We move on to show \Validity\ and \UniAg. Henceforth, 
let $i$ be a (possibly faulty) process that decides in some run of $\UOptZ$, let $m$ be the decision time of $i$, and
let $\mathtt{v}$ be the value upon which $i$ decides.

\Validity:
If $\mathtt{v}=0$, then by definition $K_i\cz$ at $m$, and so $K_i\exists0$ at $m$, and in particular $\exists0$. If $\mathtt{v}=1$, then by definition $\lnot K_i\exists0$, and so the initial value of $i$ is $1$, and so $\exists1$. Either way, we have $\existsv$ as required.

\UniAg:
It is enough to show that if $\mathtt{v}=1$, then $0$ is never decided upon in the current run.
For the rest of this proof we assume, therefore, that $\mathtt{v}=1$; therefore, by definition of $\UOptZ$, we have that both $\lnot K_i\exists0$ and no hidden path exists w.r.t.\ $\node{i,m}$. By \cref{lem:rev}, we therefore have that $K_i\notnz$ at $m$, and in particular $\notnz$ at $m$. By induction, as in the proof of \cref{solve}, we have that $\notnz$ at every time later than $m$. In particular, we have that no correct process ever learns of an initial value of $0$ (as $\notnz$ would never hold from that point on), and so $\cz$ never holds; therefore, $K_j\cz$ never holds for any $j$, and so by definition of $\UOptZ$ no process ever decides upon $0$, and the proof is complete.
\end{proof}

\begin{proof}[Proof of \cref{thm:u-opt}]

The claim follows from \cref{cor:uni-opt,u-solve}; in the boundary case of $t\!=\!0$ (which is not covered by \cref{cor:uni-opt}), we note that $\UOptZ$ and $\OptZ$ coincide, as do the problems of uniform consensus
and consensus; hence $\UOptZ$ is unbeatable, and \cref{thm:u-opt} holds, in that case as well.
\end{proof}

\begin{proof}[Proof of \cref{lem:ubeats}]

The proof has a similar structure to that of~\cref{lem:beats}. opt-EDAUC decides either one round after the sender set repeats, or at time~$\tee+1$. As argued in the proof of \cref{lem:beats}, when the sender set repeats there is a round~$k$ all of whose nodes are revealed. If they don't contain evidence of an initial value of~0, then $\UOptZ$ decides immediately. Otherwise, by \cref{lem:u-know}(a) a correct process will know~$\cz$ and decide one round later, and if this occurs at time~$m=\tee+1$, then by \cref{lem:u-know}(b)
it will decide immediately. 
An adversary~$\beta$ on which $\UOptZ$ beats opt-EDAUC with the claimed margins is a simplified version of the adversary~$\alpha$ defined  in the proof of \cref{lem:beats}. 
Denote the processes by~$\Proc=\{1,2,\ldots,n\}$. 
All initial values in~$\beta$ are~0. 
In round~1, two processes crash---process~1 and process~$2$, with process~1 sending only to process~$n$ and nobody else, and process~2 sending to everyone except process~$n$.
No process fails in round~2, and in each of the rounds $m=3,\ldots,\tee$, process~$m$ crashes without sending any messages. 
Since precisely~$\tee$ processes fail in~$\beta$ we have that $\beta\in\Crash(\tee)$.  
For $3\le m\le \tee$, every correct process fails to hear from process~$m$ in round~$m$ for the first time.  
Every correct process $i\ne n$ fails to hear from process~1 in round~1 and from process~2 in round~2, while process~$n$ fails to hear from~2 in round~1 and from process~1 in round~2. In the protocol opt-EDAUC of~\cite{CBS-uni}, no process decides before its sender set repeats, and thus all decisions are taken at time~$\tee+1$ when the adversary is~$\beta$. 
In $\UOptZ$, every correct process~$i$ sees $n-1\ge\tee+1$ values of~0 in the first round. By \cref{lem:u-know}(b) it follows that $K_i\cz$ holds at time~1, the rule for $\decideiZ$ in~$\UOptZ$ is satisfied, and  process~i decides~0 at time~1. 
\end{proof}

\subsection{Efficient Implementation of Full-Information Protocols}\label{efficient}

We now sketch the structure of communication-efficient implementations for the protocols proposed in the paper:
\begin{lemma}
\label{nlogn}
For each of the protocols $\OptZ$, $\OptMaj$, 
and $\UOptZ$ 
there is a protocol with identical decision times for all adversaries, in which every process sends at most $O(f\log n)$ bits overall to each other process. 
\end{lemma}

\begin{proof}[Proof (Sketch)]
Moses and Tuttle in \cite{MT} show how to implement full-information protocols in the crash failure model with linear-size messages. In our case, a further improvement is possible, since decisions in all of the protocols depend only on the identity of hidden nodes and on the vector of initial values. In a straightforward implementation, we can have a process $i$ report  ``{\tt value}$(j) = \valv$'' once for every $j$ whose initial value it discovers, and ``{\tt failed\_at}$(j) = \ell$'' once where~$\ell$ is the earliest failure round it knows for $j$. In addition, it should send an ``{\tt I'm\_alive}'' message in every round in which it has nothing to report. Process~$i$ can send at most one {\tt value} message and two 
{\tt failed\_at} messages for every $j$. Since {\tt I'm\_alive} is a constant-size message sent fewer than~$f+2$ times, and since encoding~$j$'s ID along with a failure round number $m\le f+2$ requires $\log n$ bits, a process~$i$ sends a total of $O(f \log n)$ bits overall. 
\end{proof}

\subsection{Different Types of Unbeatability}
\label{sec-notions}

We first formally define last-decider unbeatability.

\begin{definition}[Last-Decider Domination and Unbeatability]
\leavevmode
\begin{itemize}
\item
A decision protocol $Q$ \defemph{last-decider dominates} a protocol~$P$ in~$\gamma$, denoted by $Q\boldsymbol{\overset{\smash{l.d.}}{\dom}_\gamma} P$ if, for all adversaries $\alpha$, if $i$ the last decision in~$P[\alpha]$ is at time $m_i$, then all decisions in $Q[\alpha]$ are taken before or at $m_i$. Moreover, we say that $Q$  \defemph{strictly last-decider dominates} $P$
if $Q\overset{\smash{l.d.}}{\dom}_\gamma P$ and  $P\!\!\boldsymbol{\not}\!\!\!\overset{\smash{l.d.}}{\dom}_\gamma Q$. I.e., if for some $\alpha\in\gamma$ the last decision in $Q[\alpha]$ is {\em strictly before} the last decision in $P[\alpha]$.
\item
A protocol $P$ is a \defemph{last-decider unbeatable} solution to a decision task~$S$ in a context~$\gamma$ if $P$ solves~$S$ in~$\gamma$ and no protocol $Q$ solving~$S$ in~$\gamma$ strictly last-decider dominates~$P$.
\end{itemize}
\end{definition}

\begin{remark}
\leavevmode
\begin{itemize}
\item
If $Q\boldsymbol{\dom_\gamma} P$, then $Q\boldsymbol{\overset{\smash{l.d.}}{\dom}_\gamma} P$. (But not the other way around.)
\item
None of the above forms of strict domination implies the other.
\item
None of the above forms of unbeatability implies the other.
\end{itemize}
\end{remark}

\vspace{-3mm}
Last-decider domination does not imply domination in the sense of the rest of this paper (on which our proofs is based). 
Nonetheless, the specific property of protocols dominating $\OptZ$, $\OptMaj$, 
and $\UOptZ$, which we use to prove that these protocols are unbeatable, holds also for protocols that only last-decider dominate these protocols. 

\begin{lemma}\label{last-dom-sufficient}
\leavevmode
\begin{parts}
\item\label{last-dom-sufficient-pz}
Let $Q\overset{\smash{l.d.}}{\dom}\Pz$ satisfy \Decision. If $K_i\exists0$ at $m$ in a run $r\!=\!Q[\alpha]$ of $Q$, then $i$ decides in $r$ no later than at $m$.
\item\label{last-dom-sufficient-maj}
Let $Q\overset{\smash{l.d.}}{\dom}\OptMaj$ satisfy \Decision. If $K_i(\Maj=v)$ for $v \in \{0,1\}$ at $m$ in a run $r\!=\!Q[\alpha]$ of $Q$, then $i$ decides in $r$ no later than at $m$.
\item\label{last-dom-sufficient-upz}
Let $Q\overset{\smash{l.d.}}{\dom}\UPz$ satisfy \Decision. If $K_i\cz$ at $m$ in a run $r\!=\!Q[\alpha]$ of $Q$, then $i$ decides in $r$ no later than at $m$.
\end{parts}
\end{lemma}

The main idea in the proof of each of the parts of \cref{last-dom-sufficient} is to show that $i$ considers it possible that
all other active processes also know the fact stated in that part,
and so they must all decide by the current time in the corresponding run of the dominated protocol. Hence, the last decision in that run is made in the current time; thus, by last-decider domination, $i$ must decide.
The proofs for the first two parts are somewhat easier, as in each of these parts, any process at $m$ that sees (at least) the nodes seen by $\node{i,m}$
(or has the same initial value, if $m\!=\!0$) also knows the relevant
fact stated in that part. We demonstrate this by proving \cref{last-dom-sufficient-pz}; the analogous 
proof of \cref{last-dom-sufficient-maj} is left to the reader.

\begin{proof}[Proof of \cref{last-dom-sufficient-pz} of \cref{last-dom-sufficient}]
If $m\!=\!0$, then there exists a run $r'\!=\!Q[\beta]$ of $Q$, s.t.~~\emph{1)} $r'_i(0)\!=\!r_i(0)$,~~\emph{2)} in $r'$ all initial values are $0$, and~~\emph{3)} $i$ never fails in $r'$. Hence, in $\Pz[\beta]$ all decisions
are taken at time $m\!=\!0$, and therefore so is the last decision. Therefore, the last decision in $r'$ must be taken at time $0$. As $i$ never fails in $r'$, by \Decision\ it must decide at some
point during this run, and therefore must decide at $0$ in $r'$. As $r_i(0)\!=\!r'_i(0)$, $i$ decides at $0$ in $r$ as well, as required.

If $m\!>\!0$, then there exists a process $j$ s.t.\ $K_j\exists0$ at $m-1$ in $r$ and $\node{j,m-1}$ is seen by $\node{i,m}$. Thus, there exists a run $r'\!=\!Q[\beta]$ of $Q$,
s.t.~~\emph{1)} $r'_i(m)\!=\!r_i(m)$, and~~\emph{2)} $i$ and $j$ never fail in $r'$. Thus, all processes that are active at $m$ in $r'$ see $\node{j,m-1}$ in $r'$ and therefore know $\exists0$ in $r'$.
Hence, in $\Pz[\beta]$ all decisions are taken by time $m$, and therefore so is the last decision. Therefore, the last decision in $r'$ must be taken no later than at time $m$.
As $i$ never fails in $r'$, by \Decision\ it must decide at some point during this run, and therefore must decide by $m$ in $r'$. As $r_i(m)\!=\!r'_i(m)$, $i$ decides by $m$ in $r$ as well, as required.
\end{proof}

As the proof of \cref{last-dom-sufficient-upz} is slightly more involved, we show it as well.

\begin{proof}[Proof of \cref{last-dom-sufficient-upz} of \cref{last-dom-sufficient}]
If $m\!=\!0$, then by \cref{lem:u-know}, $t\!=\!0$. There exists a run $r'\!=\!Q[\beta]$ of $Q$, s.t.~~\emph{1)} $r'_i(0)=r_i(0)$, and~~\emph{2)} in $r'$ all initial values are $0$. Therefore,
as $t\!=\!0$, we have by \cref{lem:u-know} that all processes know $\cz$ at $m\!=\!0$ in $r'$. Hence, in $\UPz[\beta]$ all decisions
are taken at time $m\!=\!0$, and therefore so is the last decision. Therefore, the last decision in $r'$ must be taken at time $0$ as well. Since $t\!=\!0$, $i$ never fails in $r'$, and so by \Decision\ it must decide at some
point during this run, and therefore must decide at $0$ in $r'$. As $r_i(0)\!=\!r'_i(0)$, $i$ decides at $0$ in $r$ as well, as required.

If $m\!>\!0$, then there exists a process $j$ s.t.\ $K_j\exists0$ at $m\!-\!1$ in $r$ and $\node{j,m-1}$ is seen by $\node{i,m}$ in $r$.
Furthermore, as $t\!<\!n$, there exists a set of processes $I$
s.t.~~\emph{1)} $i,j \notin I$,~~\emph{2)} $|I|=t\!-\!\knownf{i,m}\!-\!1$, and~~\emph{3)} $\node{k,m\!-\!1}$ is seen by $\node{i,m}$ for every $k \in I$. Thus, there exists a run $r'=Q[\beta]$ of $Q$, s.t.~~\emph{1)} $r'_i(m)\!=\!r_i(m)$,~~\emph{2)} $i$ and $j$ never fail in $r'$,~~\emph{3)} all of $I$ fail in $r'$ at $m\!-\!1$, successfully sending messages only to $i$, and~~\emph{4)} every process at $m\!-\!1$ in $r'$ that is not seen by $\node{i,m}$, is not seen by any other process at $m$ as well. We henceforth reason about $r'$. Every process $k \ne j$ that is active at $m$ sees $\node{j,m\!-\!1}$ and furthermore satisfies
$\knownf{k,m}\ge\knownf{i,m}+|I|=t-1$. Thus, by \cref{lem:u-know}, $K_k\cz$ at $m$, and thus $k$ decides at $(\UPz[\beta],m)$.
Additionally, as $K_j\exists0$ at $m\!-\!1$, by \cref{lem:u-know} $K_j\cz$ at $m$, and thus $j$ decides at $(\UPz[\beta],m)$.
Hence, in $\UPz[\beta]$ all decisions are taken by time $m$, and therefore so is the last decision. Therefore, the last decision in $r'$ must be taken no later than at time $m$.
As $i$ never fails in $r'$, by \Decision\ it must decide at some point during this run, and therefore must decide by $m$ in $r'$. As $r_i(m)=r'_i(m)$, $i$ decides by $m$ in $r$ as well, as required.
\end{proof}

\begin{proof}[Proof of \cref{thm:last-decider}]

As explained above, \cref{thm:last-decider} follows from \cref{last-dom-sufficient}, and from the proofs of 
\cref{thm:optz,thm:optmaj,thm:u-opt}.
\end{proof}

\end{document}